\documentclass{article}
\usepackage[utf8]{inputenc}
\usepackage{amsmath,amsbsy,amssymb,array}
\usepackage{caption}
\usepackage{multirow,multicol,tabularx,booktabs}
\usepackage{graphicx,placeins,color,url}
\usepackage{amsthm}
\usepackage[a4paper]{geometry}
\usepackage{authblk}
\usepackage{algorithm} 
\usepackage{algpseudocode} 
\usepackage{pgfplots}
\usepackage{comment}
\usepackage[title]{appendix}
\RequirePackage[numbers]{natbib}
\usepackage{adjustbox}
\usepackage{graphicx}

\newtheorem{thm}{Theorem}

\newtheorem{prop}[thm]{\bf Proposition}

\newtheorem{defi}{Definition}
\newtheorem{rem}{Remark}

\def\Nc{\mbox{$\mathcal N$}}

\def\Rc{\mbox{$\mathcal R$}}

\def\Zc{{\mathcal Z}}

\def\Eb{{\mathbb E}}

\def\Pb{{\mathbb P}}

\def\Rb{\mbox{$\mathbb R$}}

\def\1{{\mathbf 1} }
\def\mds{\medskip}

\title{Risk Budgeting Portfolios: Existence and Computation}

\author{Adil Rengim CETINGOZ\thanks{Université Paris 1 Panthéon-Sorbonne, Centre d’Economie de la Sorbonne, 106 Boulevard de l’Hôpital, 75642 Paris Cedex 13, France, adil-rengim.cetingoz@etu.univ-paris1.fr.} \quad Jean-David FERMANIAN\thanks{Ensae-Crest, 5 Avenue Le Chatelier, 91120 Palaiseau, France, 
jean-david.fermanian@ensae.fr.} \quad Olivier GU\'EANT\thanks{Université Paris 1 Panthéon-Sorbonne, Centre d’Economie de la Sorbonne, 106 Boulevard de l’Hôpital, 75642 Paris Cedex 13, France, olivier.gueant@univ-paris1.fr. \textit{Corresponding author.}}}      
\date{}

\begin{document}

\maketitle
\begin{abstract}
Modern portfolio theory has provided for decades the main framework
for optimizing portfolios. Because of its sensitivity to small changes
in input parameters, especially expected returns, the mean-variance
framework proposed by Markowitz (1952) has however been challenged by new construction methods that
are purely based on risk. Among risk-based methods, the most popular
ones are Minimum Variance, Maximum Diversification, and Risk Budgeting
(especially Equal Risk Contribution) portfolios. Despite some drawbacks, Risk
Budgeting is particularly attracting because of its versatility: based
on Euler's homogeneous function theorem, it can indeed be used with a wide
range of risk measures. This paper presents mathematical
results regarding the existence and the uniqueness of Risk Budgeting portfolios for a
very wide spectrum of risk measures and shows that, for many of them,
computing the weights of Risk Budgeting portfolios only requires a standard stochastic algorithm.
\end{abstract}
\vfill
\begin{flushleft}
\end{flushleft}
\newpage

\section{Introduction}
\label{Intro}

\medskip

Seventy years ago, Markowitz \citep{markowitz1952} transformed the
financial problem of asset allocation into a simple (quadratic)
mathematical optimization problem. His modern portfolio theory has opened
the quest for quantitative tools to better invest in financial markets and construct optimized investment portfolios. Markowitz' model was
soon generalized by Tobin \cite{tobin1958liquidity} who introduced a
risk-free asset into the framework and proved what is now called a
mutual fund separation theorem. Then, \cite{lintner1965security},
\cite{mossin1966equilibrium}, \cite{sharpe1964capital} and
\cite{Treynor1962JackT} independently developed the capital asset
pricing model~(CAPM). Undoubtedly, Markowitz' modern portfolio theory,
Tobin's mutual fund separation theorem and the CAPM have shaped the
asset management industry: diversification is the cornerstone of asset
allocation; volatility is the most popular measure of portfolio risk;
so-called ``$\alpha$'' and ``$\beta$'' are ubiquitous concepts to describe strategies,
not to talk about the recurrent debates between passive and active
investing.

\medskip

Although mathematically sound, Markowitz' idea of finding the asset
allocation which maximizes the portfolio expected return given a
variance constraint -- the so-called mean-variance optimization -- raises
practical issues. In particular, choosing (or estimating) the vector
of expected returns that one inputs in the mean-variance framework is
a difficult and very sensitive task as small differences in
expected returns may yield significant differences in the final
portfolio (see \cite{best1991sensitivity}).\footnote{The case of
covariance matrices also raises issues. There is an important literature on
covariance / correlation matrix cleansing (see \cite{ledoit2004honey}
for shrinkage methods, \cite{laloux1999noise} for eigenvalue clipping,
and \cite{bun2017cleaning} for the recent rotationally invariant estimators).}

\medskip

A large set of portfolio construction methods has been proposed to
mitigate these issues such as the famous Black-Litterman model
\cite{black1992global} that uses the CAPM as a baseline for the
expected returns. Some of them simply do not rely on any expected
return. The simplest example is that of equally-weighted portfolios
(see for instance \cite{demiguel2009optimal}) which is input-agnostic
and nevertheless (or, maybe, for that reason) considered an interesting benchmark beyond
traditional capitalization-weighted ones. Most quantitative portfolio
construction methods that are independent of expected returns are in
fact risk-based methods: they focus on risk mitigation and ignore (at
least mathematically) return maximization. Most of them rely only on the
covariance matrix of asset returns. The portfolio construction method
based on finding the weights that minimize the \textit{ex ante}
variance of the portfolio return is a typical example: it corresponds
to the least risky portfolio on Markowitz' efficient frontier --
the so-called Minimum Variance portfolio. The Most-Diversified
portfolio approach is based on maximizing the diversification ratio
introduced in \cite{choueifaty2008toward}. Most-Diversified portfolios
are appealing and their properties are analyzed in
\cite{choueifaty2013properties}. A third popular risk-based approach
is Risk Budgeting, which is at the core of this paper, and
which, broadly speaking, consists in allocating risk rather than
capital in line with given ``risk budgets''. Risk Budgeting has been studied
and largely advocated by Roncalli and his coauthors in an interesting
series of papers (e.g. \cite{bruder2016risk}, \cite{bruder2012managing},
\cite{lezmi2018portfolio}, \cite{maillard2010properties}, \cite{roncalli2014introducing}) and in the
reference book \cite{roncalli2013introduction}. One of the most
adopted Risk Budgeting methods is Equal Risk Contribution (ERC) which
corresponds to the case where the risk budgets are chosen  equal. The
popularity of ERC is due to the fact that it corresponds to the
``least-concentrated'' portfolio in terms of risk and
it is relatively insensitive to small errors in the covariance matrix
estimation, compared to the Most-Diversified and Minimum-Variance
portfolios (see \cite{demey2010risk}).

\medskip

As mentioned above, Risk Budgeting focuses on the contribution of each asset
to the portfolio risk, rather than on the portfolio risk
itself. Therefore, it relies on a mathematical framework to handle the
decomposition of the total portfolio risk into asset-wise risk contributions. Such a breakdown of
risk is performed thanks to Euler's homogeneous function theorem
because of its axiomatic justification (see \cite{roncalli2013introduction} and \cite{tasche2007capital}). In particular, Risk Budgeting is
versatile and relies solely on the positive homogeneity and the
differentiability of the chosen risk measure.

\medskip

Various frameworks have been proposed to build and study Risk Budgeting portfolios. The initial framework proposed in \cite{maillard2010properties} studies ERC portfolios when volatility is the risk measure. The authors prove that whenever the covariance matrix of asset returns is positive-definite, an ERC portfolio exists and is unique. For that purpose, they regard the equations defining a Risk Budgeting portfolio as the first order condition (up to rescaling) of a constrained minimization problem.\footnote{See also \cite{bruder2012managing} for an extension beyond ERC portfolios to general Risk Budgeting portfolios.} Extensions to other risk measures are present in various papers: \cite{roncalli2014introducing} reintroduces expected returns and deals with risk measures that are a linear combination of expected return and volatility and \cite{bruder2016risk} deals with Risk Budgeting when the risk measure is Expected Shortfall (see also \cite{lezmi2018portfolio}).

\medskip

In terms of computational methods, several approaches have been proposed. The reference methods rely (up to a rescaling factor) on the numerical approximation of the solution of an unconstrained convex minimization problem, i.e. the optimization of the Lagrangian associated with the usual constrained minimization problem of the above literature. Many gradient descent algorithms have been proposed. Beyond simple methods, a Nesterov acceleration technique is applied in \cite{spinu2013algorithm} and a cyclical coordinate descent one is proposed in \cite{griveau2013fast}. 
In most cases, the reference risk measure is the volatility of portfolio returns, or at least a risk measure that may be easily computed. For instance, when asset returns are distributed according to a Gaussian mixture, \cite{lezmi2018portfolio} discusses the case of Risk Budgeting with Expected Shortfall as a way to control the ``skewness risk'' of portfolios. When the risk measure and / or the distribution of asset returns are more complex (see for instance \cite{gava2021turning}), the above methods can be challenged by simulation-based numerical methods. This was the initial motivation of this paper.

\medskip

In fact, our goal in this paper is to provide a unique framework that applies to a wide spectrum of
risk measures. For that purpose, we first revisit some existence and uniqueness results about optimal portfolios. Then, for a large range of risk measures including deviation measures (like volatility) and spectral risk measures (like Expected Shortfall), we show that building Risk Budgeting portfolios boils down to using a standard stochastic gradient descent (SGD) approach.

\medskip
In Section~\ref{RB_existence_unicity}, we introduce some notations, define the Risk Budgeting problem and formally prove the existence and the uniqueness of a Risk Budgeting portfolio. In that section, we borrow a lot from the existing literature. In particular, our proof is based on the introduction of a convex minimization problem whose first order condition is (up to rescaling) the condition that defines Risk Budgeting portfolios. In Section~\ref{Stoch_gradient_descent_risk_budgeting}, we start by summarizing the current use of Expected Shortfall in the Risk Budgeting literature and extend the semi-analytic formula of \cite{lezmi2018portfolio} to the case of Student-t mixture distributions. We then show that Expected Shortfall does not necessarily need to be computed in the definition of the convex minimization problem because Expected Shortfall has itself a variational characterization in the form of a minimum: the famous Rockafellar-Uryasev formula (see \cite{rockafellar2000optimization}, for instance). In particular, for almost any distribution of asset returns, a stochastic gradient descent enables to compute the Risk Budgeting portfolio when Expected Shortfall is the risk measure. Section~\ref{Extending_framework} generalizes the approach and shows that the set of risk measures for which building a Risk Budgeting portfolio boils down to using a standard stochastic gradient descent is very large. Numerical examples are provided and discussed in Section~\ref{Numerical_results}.

\section{Risk Budgeting: definition, existence and uniqueness}
\label{RB_existence_unicity}

\subsection{Notations and statement of the problem}

Let us first define a probability space $(\Omega,\mathcal{F}, \mathbb{P})$. This probability space will be used throughout the paper and we denote by $L^0(\Omega,\mathbb{R}^d)$ the set of  $\mathbb R^d$-valued random variables, i.e. measurable functions defined on $\Omega$ with values in $\mathbb R^d$ ($d \ge 1$).

\medskip

Hereafter, risk measures are functions mapping a random variable (regarded as a loss) to a real number assessing the risk of this random variable. To be compatible with our analysis of the Risk Budgeting problem, we require that they satisfy two classical properties, as defined below:

\begin{defi}
A function $\rho : L^0(\Omega,\mathbb{R}) \longrightarrow \mathbb{R}$ is 
said to be an RB-compatible risk measure if it satisfies the following assumptions:
\begin{eqnarray*}
\textbf{\textup{(PH)}} & \forall Z \in  L^0(\Omega,\mathbb{R}) , \forall \lambda\geq0 , \quad \rho(\lambda Z) = \lambda \rho(Z) & \text{(positive homogeneity)}\\
\textbf{\textup{(SA)}} & \forall Z_1,Z_2  \in  L^0(\Omega,\mathbb{R}) , \quad \rho(Z_1+Z_2) \leq \rho(Z_1)+\rho(Z_2) & \text{(sub-additivity)}.
\end{eqnarray*}
\end{defi}

As a consequence, an RB-compatible risk measure is convex. 

\begin{rem}
Most risk measures are in fact defined on a subset of  $L^0(\Omega,\mathbb{R})$, like $L^1(\Omega,\mathbb{R})$ or $ L^2(\Omega,\mathbb{R})$.
\end{rem}

\begin{rem}
In risk management, a classical concept is that of coherent risk measures (see~\cite{artzner1999coherent}). In addition to the two properties imposed in the above definition, coherence requires other properties:
\begin{eqnarray*}
 & \forall Z_1,Z_2  \in  L^0(\Omega,\mathbb{R}), \quad Z_1 \le Z_2  \implies \rho(Z_1) \leq \rho(Z_2)  & \text{(monotonicity)}\\
 & \forall (Z,c) \in  L^0(\Omega,\mathbb{R}) \times \mathbb{R}, \quad \rho(Z+c) = \rho(Z) + c  & \text{(translation invariance)}.
\end{eqnarray*}
 Coherent risk measures are RB-compatible risk measures but the converse is not true. In particular, volatility is an RB-compatible risk measure but not a coherent risk measure (see the discussion in Section \ref{Extending_framework}).
\end{rem}

\medskip

In order to define Risk Budgeting portfolios, let us consider a financial universe with $d$ assets. Their returns are stacked in a random vector $X$ with values in $\mathbb R^d$, i.e. $X \in L^0(\Omega,\mathbb{R}^d)$.
Portfolios are identified with vectors of weights, hereafter denoted by $\theta$, which belong to the simplex $\Delta_d = \{ \theta \in \mathbb{R}_+^d | \theta_1 + \ldots +\theta_d = 1\}$. 
To any $d$-dimensional random vector $X$ and any RB-compatible risk measure $\rho$, 
we can associate a function $\mathcal{R}_{\rho, X}$ (hereafter denoted by $\mathcal{R}$ when there is no ambiguity) defined by
\begin{alignat*}{2}
\mathcal{R}:\mathbb{R}^d_+ &\longrightarrow \mathbb{R} \\
y&\longmapsto\rho(-y'X).
\end{alignat*}
If the portfolio weights are $\theta \in \Delta_d$, then $\theta' X$ is the return of the portfolio and $\Rc (\theta)$ is the risk associated with this portfolio.
In what follows, we shall always assume that $\mathcal{R}$ is continuous on $\mathbb{R}_+^d$ and continuously differentiable on $(\mathbb{R}^*_+)^d$. 

\medskip

Because of the positive homogeneity of the risk measure $\rho$ and as an application of Euler's homogeneous function theorem, the risk associated with a portfolio represented by its vector of weights $\theta$ can be decomposed as
$\mathcal{R}(\theta) = \sum_{i=1}^d \theta_i \partial_i \mathcal{R}(\theta)$, where $\partial_i \mathcal{R}(\theta)$ is the partial derivative of $\mathcal{R}$ with respect to the variable~$\theta_i$ at point $\theta$. In the Risk Budgeting literature, the $i^{\text{th}}$ term of the above sum is called the risk contribution of asset $i$ to the total risk of the portfolio. The Risk Budgeting problem consists then in finding a vector of weights $\theta$ for which risk contributions are equal to desired proportions (represented by a vector of risk budgets $b$) of the total risk.\\ 
\begin{defi}
Let $\Delta^{>0}_d = \{ \theta \in (\mathbb{R}^*_+)^d | \theta_1 + \ldots +\theta_d = 1\}$ and let $b \in \Delta^{>0}_d$ be a vector of risk budgets. We say that a vector of weights $\theta \in \Delta^{>0}_d$ solves the Risk Budgeting problem $\text{RB}_{b}$ if and only if 
$$\quad \theta_i \partial_i \mathcal{R}(\theta) = b_i \mathcal{R}(\theta),$$ 
for every $i \in \{1,\ldots,d\}$.
\end{defi}

\medskip

Hereafter, we assume that $\mathcal{R}(\theta) > 0$ for any $\theta \in\Delta_d$, i.e., the risk of any long-only portfolio is positive.
This restriction does not mean that we restrict ourselves to loss-based risk measures (see~\cite{contduguestcoignard}). For risk measures like Expected Shortfall, it just means that the level is sufficiently high to prevent negative values.  This assumption is necessary to prove our main result (Theorem~\ref{thm_existence} below)  and makes sense in the context of Risk Budgeting and positive weights, where the overall risk of a portfolio has to be distributed across its individual components.

\subsection{Theoretical results on Risk Budgeting}

Given the above definition of Risk Budgeting portfolios, two questions naturally arise:
\begin{enumerate}
    \item the existence of a vector of weights $\theta$ that solves $\text{RB}_{b}$ for a vector of risk budgets $b \in \Delta_d^{>0}$;
    \item the uniqueness of such a vector of weights $\theta$. 
\end{enumerate}
The following theorem solves the first point. 
\begin{thm}
\label{thm_existence}
Let $b \in \Delta_d^{>0}$.
Let $g : \mathbb{R_+} \longrightarrow \mathbb{R}$ be a continuously differentiable convex and increasing function.
Let the function $\Gamma_g:(\Rb_+^*)^d \rightarrow \Rb$ be defined by 
$$ \Gamma_g: y \longmapsto g\big(\mathcal{R}(y)\big) - \sum_{i=1}^d b_i \log{y_i}.$$
There exists a unique minimizer $y^*$ of the function $\Gamma_g$ and $\theta^* := \frac{y^*}{\sum_{i=1}^d y^*_i} \in \Delta_d^{>0}$ solves $\text{RB}_{b}$.
\end{thm}

\begin{proof}
Since $\mathcal{R}(\theta) > 0$ for every $\theta \in \Delta_d^{>0}$ by assumption, \textbf{(PH)} implies 
$\mathcal{R}(y) > 0$ for every $ y \in \mathbb{(R_+^*)}^d$. Therefore, $\Gamma_g(y) = g\big(\mathcal{R}(y)\big) - \sum_{i=1}^d b_i \log{y_i}$ is well defined for all $y \in (\Rb_+^*)^d$.

\medskip
Let us then notice that $\Gamma_g$ is strictly convex  since $g$ is convex and increasing, $\mathcal{R}$ is convex, 
and $y \in (\Rb_+^*)^d \longmapsto - \sum_{i=1}^d b_i \log{y_i}$ is strictly convex.

\medskip
To prove the existence of a minimizer to the function $\Gamma_g$, for any $\theta \in \Delta_d^{>0}$, let us introduce the function $\gamma_{g,\theta}:\mathbb{R_+^*} \longrightarrow \mathbb{R}$ defined by
$$ \gamma_{g,\theta}: \lambda \longmapsto \Gamma_g(\lambda\theta) = g(\lambda\mathcal{R}(\theta)) - \sum_{i=1}^d b_i \log{\theta_i} - \log\lambda.$$

We first notice that $\lim_{\lambda \to 0^+} \gamma_{g,\theta}(\lambda) = \lim_{\lambda \to +\infty} \gamma_{g,\theta}(\lambda) = +\infty$.\footnote{For the latter point, note that the convexity of $g$ implies $g(\lambda\mathcal{R}(\theta)) \geq g(c) + g'(c)(\lambda\mathcal{R}(\theta) - c)$ for any $c$ such that $g'(c) > 0$ (and there exists such a $c$ since $g$ is continuously differentiable and strictly increasing).} By continuity, there exists $\lambda^*(\theta)$ such that $\gamma_{g,\theta}(\lambda) \geq \gamma_{g,\theta}(\lambda^*(\theta))$ for every $\lambda > 0$. Let us show by contradiction that $\theta\longmapsto \lambda^*(\theta)$ is bounded.

\medskip 

For that purpose, assume the existence of a sequence $(\theta_n)_n$ with values in $\Delta_d^{>0}$ such that $\lambda_n := \lambda^*(\theta_n) \to +\infty$. We can then extract a subsequence $(\theta_{\varphi(n)})_n$ that converges towards $\bar\theta \in \Delta_d$ and such that $\lambda_{\varphi(n)} \to +\infty$. 
For all $n$, $\lambda_n$ satisfies the first order condition $\gamma_{g,\theta_{\varphi(n)}}'(\lambda_{\varphi(n)}) = 0$, i.e.
$$\mathcal{R}(\theta_{\varphi(n)}) g'\big(\lambda_{\varphi(n)}\mathcal{R}(\theta_{\varphi(n)})\big) - \frac{1}{\lambda_{\varphi(n)}} = 0.$$
Therefore, if $x_n := \lambda_{\varphi(n)}\mathcal{R}(\theta_{\varphi(n)})$, then we have $x_ng'(x_n)=1$ for all $n$. However, because $\lim_{n \to +\infty}\lambda_{\varphi(n)} = +\infty$ and $\lim_{n \to +\infty}\mathcal{R}(\theta_{\varphi(n)}) = \mathcal{R}(\bar\theta) > 0$, we have $\lim_{n \to +\infty} x_n = +\infty$. This contradicts $x_ng'(x_n)=1$ for all $n$, because $g$ is a convex and increasing function.
Therefore, we have proved that $\theta\longmapsto \lambda^*(\theta)$ is bounded: there exists a constant $M$ such that $\lambda^*(\theta) \leq M$.

\medskip
For every $ y \in \mathbb{(R_+^*)}^d$, we have
$$\Gamma_g(y) = \gamma_{g,y/\sum_i y_i}\Big(\sum_i y_i\Big) \geq \gamma_{g,y/\sum_i y_i}\bigg(\lambda^*\Big(\frac{y}{\sum_i y_i}\Big)\bigg) = \Gamma_g\bigg(\frac{y}{\sum_i y_i}\lambda^*\Big(\frac{y}{\sum_i y_i}\Big)\bigg) .$$
Setting $C_M := \Big\{y \in \mathbb{(R_+^*)}^d | \sum_{i=1}^d y_i \leq M \Big\}$, we deduce
$$\inf_{y \in \mathbb{(R_+^*)}^d} \Gamma_g(y) = \inf_{y \in C_M} \Gamma_g(y).$$

\medskip
Now, let us consider an arbitrary vector $\bar{y} \in C_M$ and define  
$$\varepsilon := \min\bigg(\min_{i} \bar{y_i}, \min_i\exp\Big(\frac{1}{b_i}\big(g(0) - (1-b_i)\log{M} - \Gamma_g(\bar{y})\big)\Big)\bigg).$$
For any $y \in C_M$, if there exists $j \in \big\{1,\dots,d\big\}$ such that $y_j<\varepsilon$, then, by definition of $\varepsilon$,
\begin{align*}
\Gamma_g(y) &= g\big(\mathcal{R}(y)\big) - \sum_{i=1}^d b_i \log{y_i} \geq g(0) - \sum_{i=1}^d b_i \log{y_i}\\
&\geq g(0) - \sum_{i\neq j} b_i \log{y_i} - b_j\log\varepsilon
\geq g(0) - \log M \sum_{i\neq j} b_i  - b_j\log\varepsilon\\
&\geq g(0) - \log M (1 - b_j) - b_j\log\varepsilon
\geq \Gamma_g(\bar{y}).
\end{align*}

\medskip

Setting $D_\varepsilon := [\varepsilon,+\infty)^d$, we have 
therefore $$\inf_{y \in C_M} \Gamma_g(y) = \inf_{y \in C_M\cap D_\varepsilon} \Gamma_g(y).$$
As the nonempty\footnote{It contains $\bar{y}$ by definition of $\varepsilon$.} set $C_M \cap D_\varepsilon$ is compact, there exists $y^* \in C_M \cap D_\varepsilon$ such that 
$ \Gamma_g(y) \geq  \Gamma_g(y^*)$ for every $y \in C_M \cap D_\varepsilon$. 
We deduce that $ \Gamma_g(y) \geq  \Gamma_g(y^*)$ for every 
$y \in \mathbb{(R_+^*)}^d$, proving the first assertion of the theorem. 

\medskip
The uniqueness of the minimizer is a consequence of the strict convexity of $\Gamma_g$.
\medskip

Now, as $y^*$ is an interior minimum of $\Gamma_g$, we have

$$\forall i \in \big\{1,\ldots,d\big\}, \quad g'(\mathcal{R}(y^*)) \partial_i \mathcal{R}(y^*) - \frac{b_i}{y_i^*} = 0, $$
or, equivalently,
$$\forall i \in \big\{1,\ldots,d\big\}, \quad y_i^* g'(\mathcal{R}(y^*)) \partial_i \mathcal{R}(y^*) = b_i.$$
Summing over $i\in \{1,\ldots,d\}$, Euler's homogeneous function theorem gives
$\mathcal{R}(y^*) g'\big(\mathcal{R}(y^*)\big) = 1$. Therefore, we get
$$\forall i \in \big\{1,\ldots,d\big\}, \quad
    y_i^* \partial_i \mathcal{R}(y^*) = b_i \mathcal{R}(y^*).$$
Setting $\theta^* = y^*/\sum_{i=1}^d y^*_i$ and using \textbf{(PH)}, we see that $\theta^*$ solves $\text{RB}_{b}$.
\end{proof}

After having proven the existence of a solution to the Risk Budgeting problem, let us deal with uniqueness. 
\begin{thm}
\label{thm_uniqueness}
Let $b \in \Delta^{>0}_d$ and let $\theta \in \Delta^{>0}_d$ be a solution of $\text{RB}_b$.
Let $g : \mathbb{R_+} \rightarrow \mathbb{R}$ be a continuously differentiable convex and increasing function.
Consider the map $\Gamma_g$ as defined in Theorem~\ref{thm_existence}, and let $y^*$ be the unique minimizer of $\Gamma_g$.\\
Then, we have $$\theta = \frac{y^*}{\sum_{i=1}^d y^*_i}\cdot$$
\end{thm}

\begin{proof}
The function $h: \lambda \in \mathbb R_+ \mapsto \mathcal{R}(\lambda\theta) g'(\mathcal{R}(\lambda\theta))$ is continuous because $\mathcal{R}$ and $g'$ are continuous. Since $h(0)=0$ and $\lim_{\lambda \to +\infty} h(\lambda) = +\infty$, there exists $\bar{\lambda}\in \mathbb R_+$ such that $h(\bar{\lambda})=1$.

\medskip

Defining $y := \bar{\lambda} \theta$, we obtain, for all $i \in \{1,\ldots,d\}$,
$$y_i \partial_i \mathcal{R}(y) g'\big(\mathcal{R}(y)\big) = \bar{\lambda} \theta_i \partial_i \mathcal{R}(\bar{\lambda} \theta) g'\big(\mathcal{R}(\bar{\lambda} \theta)\big) =
\frac{\bar{\lambda} \theta_i \partial_i \mathcal{R}(\bar{\lambda} \theta)}{ \mathcal{R}(\bar{\lambda}\theta)}= \frac{\bar{\lambda} \theta_i \partial_i \mathcal{R}( \theta)}{ \bar{\lambda}\mathcal{R}(\theta)}  = \frac{\theta_i \partial_i \mathcal{R}( \theta)}{ \mathcal{R}(\theta)} = b_i,$$
because $\theta$ is a solution of $\text{RB}_b$ and $\mathcal{R}$ (resp. $\partial_i \mathcal{R}$) is positively homogeneous of degree~$1$ (resp. of degree $0$).
In other words, this yields
$$\quad \partial_i \mathcal{R}(y) g'\big(\mathcal{R}(y)\big) - \frac{b_i}{y_i} = 0, \forall i \in \{1,\ldots,d\}, $$ 
and $y$ is a critical point of the convex function $\Gamma_g$.
\medskip
We conclude that $y^* = y = \bar{\lambda} \theta$. Since $\theta \in \Delta_d^{>0}$, we must have
$\theta = y^*/\sum_{i=1}^d y^*_i$.
\end{proof}

The above theorems prove the existence of a unique solution to the Risk Budgeting problem for any vector of positive budgets and provide a variational characterization (up to rescaling) of the vector of weights. The link between the Risk Budgeting problem and the unconstrained convex minimization problem we defined in Theorem~\ref{thm_existence} was already noticed by several authors, notably \cite{bruder2012managing,griveau2013fast,roncalli2013introduction,spinu2013algorithm}, when $g$ is the identity map (a sketch of the proof of Theorems~\ref{thm_existence} and~\ref{thm_uniqueness} has for instance been proposed in \cite{roncalli2013introduction}, Section 2.2.2.2). Nonetheless, to the best of our knowledge, no complete mathematical proof of the existence and uniqueness of Risk Budgeting portfolios has been stated in the literature until now. Moreover, the possibility of choosing functions $g$ beyond the identity function will prove to be useful for developing our general framework in the next sections.

\begin{rem}
The above theorems are also useful to shed light on the advantage of Risk Budgeting for building portfolios. Indeed, specifying risk budgets rather than weights allows somehow to reduce risk as it has been shown in \cite{roncalli2013introduction}. Indeed, for any $b \in \Delta^{>0}_d$, if $\theta \in \Delta^{>0}_d$ is the solution of $\text{RB}_b$, then $\mathcal{R}(\theta) \le \mathcal{R}(b).$\footnote{If $y^*$ is the minimizer of~$y \mapsto \mathcal{R}(y)\! -\! \sum_{i=1}^d b_i \log{y_i}$ on $\mathbb{(R_+^*)}^d$ and
 $\lambda := \sum_{i=1}^d y^*_i$, then we have seen that $\theta=y^*/\lambda$ and
\begin{eqnarray*}
\mathcal{R}(\theta) &=& \frac{1}{\lambda} \mathcal{R}(y^*) \le \frac{1}{\lambda} \bigg(\mathcal{R}(\lambda b)  - \sum_{i=1}^d b_i \log(\lambda b_i) + \sum_{i=1}^d b_i \log{y^*_i}\bigg)\\
&\le& \mathcal{R}(b) - \frac 1\lambda\sum_{i=1}^d b_i \log\bigg(\frac{b_i}{\theta_i}\bigg) \le \mathcal{R}(b),
\end{eqnarray*}
because relative entropy is non-negative.}
\end{rem}

\section{Risk Budgeting with Expected Shortfall}
\label{Stoch_gradient_descent_risk_budgeting}

\subsection{The importance of Expected Shortfall}
\label{importance_Expected_Shortfall}
In the first papers advocating for the Risk Budgeting approach, the chosen risk measure was volatility. Volatility indeed constitutes a reasonable choice of risk measure, especially when the probability distributions of asset returns do not exhibit asymmetry and / or heavy tails.

\medskip

Denoting by $\Sigma$ the covariance matrix of asset returns, the volatility of the portfolio defined by the vector of weights $y$ is obviously $\mathcal{R}(y) := \sqrt{y'\Sigma y}$. The latter quantity is easily computed and the Risk Budgeting problem can be efficiently solved using a gradient descent procedure,\footnote{When $d=2$, it is possible to derive a closed-form  solution of the Risk Budgeting problem for volatility. In the general case ($d>2$), analytical solutions only exist in very specific cases, notably if we assume that the correlations between asset returns are all equal to the same constant in $\left\{\frac{-1}{d-1},0,1\right\}$. If we consider the specific case of ERC (i.e. $b_i=\frac{1}{d}$) for volatility, then a closed-form solution is known when all correlations are equal and in that case the ERC portfolio has weights inversely proportional to volatilities. See Sections 2.2.2.1, 2.3.1 and 2.3.2 in~\cite{roncalli2013introduction} for an in-depth analysis of these specific cases.} to minimize over $\mathbb{(R_+^*)}^d$ the function  
\begin{equation*}
 \Gamma_{x \mapsto x^2} : y \mapsto \big(\mathcal{R}(y)\big)^2 - \sum_{i=1}^d b_i \log{y_i} = y'\Sigma y - \sum_{i=1}^d b_i \log{y_i}. 
\end{equation*} 

\medskip

However, it is well-known that asset and portfolio returns exhibit skewed and heavy-tailed distributions. 
And numerous studies show that excess returns reward investors for carrying the risk of sudden and significant losses \cite{lemperiere2017risk, rietz1988equity}.
Therefore, to more efficiently deal with such distributional features in portfolio management, it makes sense to use other risk measures.

\medskip

The most classical risk measure in finance beyond volatility is Value-at-Risk. The Value-at-Risk at level $\alpha \in (0,1)$ of a real-valued random variable $Z$ (regarded as a loss) is $$\text{VaR}_\alpha(Z) := \inf \{ z \in \mathbb R  \, | \, \mathbb P(Z \le z) \ge \alpha \}.$$ In the case of a continuous random variable with positive density, $\text{VaR}_\alpha(Z)$ is uniquely characterized by $\mathbb{P}\big(Z \leq \text{VaR}_\alpha(Z)\big) = \alpha$ or equivalently $\mathbb{P}\big(Z \geq \text{VaR}_\alpha(Z)\big) = 1 - \alpha$: it corresponds to a threshold of loss that is exceeded with probability $1-\alpha$. Value-at-Risk is widely used by practitioners and regulators but it suffers from the major drawback of not being sub-additive. From this perspective, a better risk measure is Expected Shortfall. The Expected Shortfall at level $\alpha \in (0,1)$ of a real-valued $L^1$ random variable $Z$ (regarded as a loss) is defined by $$\text{ES}_\alpha(Z) := \frac 1{1-\alpha} \int_{\alpha}^1 \,\text{VaR}_s(Z) \,ds.$$ For a continuous and $L^1$ real-valued random variable $Z$ with positive density, Expected Shortfall at level $\alpha \in (0,1)$ can be written as $\text{ES}_\alpha(Z) = \mathbb{E} \big[Z | Z \geq \text{VaR}_\alpha(Z)\big]$ and it corresponds therefore to the average loss beyond the Value-at-Risk threshold. Expected Shortfall is a coherent risk measure (see \cite{acerbi2002coherence}) and it is therefore an RB-compatible risk measure.

\medskip

Expected Shortfall has been considered in the recent Risk Budgeting literature. For instance,~\cite{lezmi2018portfolio} proposed to use Expected Shortfall to construct Risk Budgeting portfolios 
because the latter risk measure allows to focus on the left tail of P\&L distributions only, contrary to volatility.\footnote{Unlike volatility, Expected Shortfall depends on expected returns. Nonetheless, it is always possible to translate random variables by their expectations to capture skewness risk independently of expected returns (see Section \ref{Extending_framework}).}

\medskip

When the chosen risk measure is Expected Shortfall, there is no simple formula for $\mathcal{R}(y)$ in general. There exist nevertheless some cases in which Expected Shortfall is easily computed. For instance, when asset returns are distributed according to a mixture of two Gaussian distributions as in \cite{lezmi2018portfolio},
there exist semi-analytic expressions for Expected Shortfalls. Then, the above gradient procedure works fine because it can rely on semi-analytic formulas.

\medskip

In this section, we first show that some semi-analytic expressions for Expected Shortfall are available whenever the underlying asset returns are distributed according to a Student-t mixture. Therefore, we extend the scope of the probability distributions for which the above gradient descent procedure should be the reference method for solving Risk Budgeting problems. Note that we must choose a level $\alpha$ so that the Expected Shortfall (at that level) of all long-only portfolios is positive. Then, we propose a general method that does not rely on any parametric assumption for the joint law of the asset returns, but it requires to use a stochastic gradient descent rather than a simple gradient descent. 

\subsection{A parametric model with semi-analytic expressions}
\label{parametric_model_analytic}

There are many possible choices for modeling the joint distribution of asset returns. Because asset returns often exhibit skewed and heavy-tailed distributions, Student-t mixtures are natural candidates. Student-t distributions indeed generate heavy tails, and mixing them allows to model skewness when some of them are not centered. Now, we show that Expected Shortfall can be computed in a semi-analytic manner when the vector of asset returns $X$ is a mixture of~$N$~multivariate Student-t distributions. This means $X$ can be written as
$$X = 1_{C = 1} X_1 + \ldots + 1_{C = N} X_N, $$ 
where:
\begin{itemize}
    \item for all $i \in \{1,\ldots,N\}$, $X_i$ follows a $d$-dimensional Student-t distribution $t(\mu_i, \Lambda_i, \nu_i)$ with location parameter $\mu_i$, positive-definite scale matrix $\Lambda_i$, and $\nu_i>1$ degrees of freedom, i.e. $X_i$ has a density $$f(x|\mu_i, \Lambda_i, \nu_i) := \frac{\Gamma\left[(\nu_i+d)/2\right]}{\Gamma(\nu_i/2)\nu_i^{d/2}\pi^{d/2}\text{det}(\Lambda_i)^{1/2}}\left\{1+\frac{1}{\nu_i}({x}-{\mu_i})'{\Lambda_i}^{-1}({ x}-{\mu_i})\right\}^{-(\nu_i+d)/2};$$
    \item $C$ is a discrete random variable with values in $\{1, \ldots, N\}$ and $\Pb(C=i) = p_i$ for all $i\in \{1,\ldots,N\}$, $\sum_{i=1}^N p_i = 1$;
    \item $C, X_1, \ldots, X_N$ are mutually independent.
\end{itemize}

For $y \in \mathbb R^d$, the probability density function and the cumulative distribution function of the loss $-y'X$ are respectively  
$$ f_{-y'X}:z \mapsto \sum\limits_{i=1}^N  \frac{p_i}{\sqrt{y'\Lambda_iy}} f_{t_{\nu_i}} \left(\frac{z + y'\mu_i}{\sqrt{y'\Lambda_iy}}\right)$$
and
$$ F_{-y'X}: z \mapsto \sum\limits_{i=1}^N  \frac{p_i}{\sqrt{y'\Lambda_iy}} F_{t_{\nu_i}} \left(\frac{z + y'\mu_i}{\sqrt{y'\Lambda_iy}}\right),$$
where $f_{t_{\nu_i}}$ and $F_{t_{\nu_i}}$ denote respectively the density function and the cumulative distribution function of a standard Student-t distribution with $\nu_i$ degrees of freedom.

\medskip

Note that the latter distributions are continuous. Therefore, the Value-at-Risk at level $\alpha \in (0,1)$ associated with the loss $-y'X$ is characterized by the relationship
\begin{equation*}
     \sum\limits_{i=1}^N  p_i F_{t_{\nu_i}} \Bigg(\frac{\text{VaR}_\alpha(-y'X) + y'\mu_i}{\sqrt{y'\Lambda_iy}}\Bigg) = \alpha.
\end{equation*}

The associated Expected Shortfall at level $\alpha$ is then
\begin{align*}
    &\text{ES}_\alpha(-y'X)
    = \frac{1}{1-\alpha} \sum\limits_{i=1}^N  \frac{p_i}{\sqrt{y'\Lambda_iy}} \int_{\text{VaR}_\alpha(-y'X)}^{+\infty} z f_{t_{\nu_i}} \bigg(\frac{z + y'\mu_i}{\sqrt{y'\Lambda_iy}}\bigg) \, dz \nonumber\\
    &= \frac{1}{1-\alpha} \sum\limits_{i=1}^N  p_i \int_{\frac{\text{VaR}_\alpha(-y'X) + y'\mu_i}{\sqrt{y'\Lambda_iy}}}^{+\infty} (u{\sqrt{y'\Lambda_iy}} - y'\mu_i) f_{t_{\nu_i}} (u)\, du \nonumber\\
    &= \frac{1}{1-\alpha} \sum\limits_{i=1}^N  p_i \Bigg\{ {\sqrt{y'\Lambda_iy}}\int_{\frac{\text{VaR}_\alpha(-y'X) + y'\mu_i}{\sqrt{y'\Lambda_iy}}}^{+\infty} u f_{t_{\nu_i}} (u)\, du  - 
    y'\mu_i \int_{\frac{\text{VaR}_\alpha(-y'X) + y'\mu_i}{\sqrt{y'\Lambda_iy}}}^{+\infty} f_{t_{\nu_i}} (u)\, du \Bigg\}\nonumber\\
    &= \frac{1}{1-\alpha} \sum\limits_{i=1}^N  p_i \Bigg\{{\sqrt{y'\Lambda_iy}}\int_{\frac{\text{VaR}_\alpha(-y'X) + y'\mu_i}{\sqrt{y'\Lambda_iy}}}^{+\infty} u f_{t_{\nu_i}} (u) _,du  - 
    y'\mu_i \int_{-\infty}^{-\frac{\text{VaR}_\alpha(-y'X) + y'\mu_i}{\sqrt{y'\Lambda_iy}}} f_{t_{\nu_i}} (u)\, du \Bigg\}\nonumber\\
    &= \frac{1}{1-\alpha} \sum\limits_{i=1}^N  p_i \Bigg\{ {\sqrt{y'\Lambda_iy}} \bigg(\frac{\nu_i+\left(\frac{\text{VaR}_\alpha(-y'X) + y'\mu_i}{\sqrt{y'\Lambda_iy}}\right)^2}{\nu_i-1}\bigg) f_{t_{\nu_i}} \left(\frac{\text{VaR}_\alpha(-y'X) + y'\mu_i}{\sqrt{y'\Lambda_iy}}\right)  \nonumber\\
    &\ \ \ - y'\mu_i F_{t_{\nu_i}} \bigg(-\frac{\text{VaR}_\alpha(-y'X) + y'\mu_i}{\sqrt{y'\Lambda_iy}}\bigg) \Bigg\} \nonumber
\end{align*}
where we used the identity $\int_t^{+\infty} u f_{t_\nu}(u) du = (\nu +t^2) f_{t_\nu}(t) /(\nu-1)$, obtained by a direct calculation.

\medskip

Given that $\text{VaR}_\alpha(-y'X)$ can easily be computed with a root-solving algorithm,
the above expression can be regarded as semi-analytic. In particular, following Theorem~\ref{thm_existence}, a simple gradient descent procedure can be used to compute Risk Budgeting portfolios in this framework. 

\subsection{Towards a stochastic optimization problem}
\label{tsop}
Although some parametric models as in Section~\ref{parametric_model_analytic} could be used, they can fall short of being a good representation of asset returns. For that reason, more complex models might be preferred despite the lack of a semi-analytic expression for Expected Shortfall. For example, 
in order to capture joint tails, \cite{gava2021turning} proposed some Pareto distributions for the left tail of individual asset returns and a Vine copula to model the dependence structure. For such a setting, in order to use a gradient descent procedure, one needs to estimate the Expected Shortfall term at each step of the optimization algorithm. This typically requires a computer-intensive approach as Monte Carlo simulations are usually necessary to estimate Value-at-Risk and then Expected Shortfall.  

\medskip

An alternative route consists in estimating Expected Shortfall and computing the solution of the Risk Budgeting problem at the same time. This route is based on the variational characterization of Expected Shortfall known as the Rockafellar-Uryasev formula (see \cite{rockafellar2002conditional}):
\begin{equation*}
\text{ES}_\alpha(Z) = \inf_{\zeta \in \mathbb R} \Big( \zeta + \frac 1{1-\alpha}  \Eb\big[(Z-\zeta)_+\big] \Big),
\end{equation*}
for any real-valued random variable $Z\in L^1$ (regarded as a loss). Moreover, the infimum in the above formula is in fact a minimum and, when $Z$ has positive density, the minimizer is unique, given by $\text{VaR}_\alpha(Z)$.

\medskip 
Using Rockafellar-Uryasev formula, the function $\Gamma_{\text{Id}}$ in Theorem~\ref{thm_existence} is written
$$\Gamma_{\text{Id}} : y \mapsto \min_{\zeta \in \mathbb R} \mathbb E\Big[ \zeta + \frac 1{1-\alpha} (-y'X-\zeta)_+\Big] - \sum_{i=1}^d b_i \log{y_i}.
$$
Therefore, solving $\text{RB}_b$ boils down to solving the stochastic optimization problem
\begin{equation}
\label{stocRB}
    \min_{(y,\zeta) \in (\mathbb R_+^*)^d\times\mathbb R} \mathbb E\Big[ \zeta + \frac 1{1-\alpha} (-y'X-\zeta)_+ - \sum_{i=1}^d b_i \log{y_i}\Big].
\end{equation}
More precisely, once the solution $(y^*,\zeta^*)$ of the stochastic optimization problem has been found, typically using a stochastic gradient descent, we simply need to normalize $y^*$ and $\theta := y^*/\sum_{i=1}^d y^*_i$ solves $\text{RB}_b$.

\medskip

By using the above variational characterization, we have seen that solving the Risk Budgeting problem with Expected Shortfall is reduced to solving a stochastic optimization problem. 
It must be noted that this approach can be applied to a large range of asset return distributions since the only requirement is that of a finite mean. 
This approach is therefore, somehow, universal in terms of asset return distributions. 

\medskip
Is this approach specific to Expected Shortfall? We tackle this question in the next section and show that similar ideas can be used for a large set of risk measures. 

\section{Extension to other RB-compatible risk measures}
\label{Extending_framework}

\subsection{Introduction and preliminary remarks}

Consider an RB-compatible risk measure $\rho$. The ideas outlined in Section~\ref{tsop} formally apply when 
there exists a continuously differentiable, convex and increasing function $g : \mathbb{R_+} \rightarrow \mathbb{R}$ such that 
\begin{equation*}
g\big(\mathcal{R}(y)\big) = g\big(\rho(-y'X)\big) = \min_{\zeta \in \mathcal{Z}} \mathbb E\big[H(\zeta, -y'X) \big]
\end{equation*}
for some set $\mathcal{Z}$ and some function $H$.\footnote{Of course, we must choose a risk measure that is positive for all long-only portfolios, otherwise our Risk Budgeting problem does not make sense.} 
Indeed, in such a case, the $\text{RB}_b$ problem boils down to the stochastic optimization problem 
\begin{equation*}
  \min_{y \in (\mathbb{R}_+^*)^d, \zeta \in \mathcal Z} \mathbb E\Big[H(\zeta, -y'X) - \sum_{i=1}^d b_i \log{y_i}\Big].
\end{equation*}
Like in Section~\ref{tsop}, if $(y^*,\zeta^*)$ is a solution of the above stochastic optimization problem, then $\theta := y^*/\sum_{i=1}^d y^*_i$ solves $\text{RB}_b$.

\medskip

It is noteworthy that the theorems of Section \ref{RB_existence_unicity} involve a function $g$. In the case of Expected Shortfall, we used the identity function, i.e. $g(x) = x$, and that case deserves several remarks.

\medskip

Risk measures that can be characterized as minimizers have significantly attracted the attention of the academic literature (see \cite{gneiting2011making} and the concept of elicitability). In spite of their attractiveness for optimization problems, risk measures characterized by minima have been less studied. An important paper dealing with risk measures in the form of minima is \cite{rockafellar2013fundamental} in which the authors introduce a general framework involving the famous quadrangles made of error, deviation, regret and risk. In their setting, risk measures are minima and even minima of expected values in the case of expectation quadrangles. In their interesting paper \cite{embrechts2021bayes}, Embrechts and his coauthors introduced the notions of Bayes pair and Bayes risk measure which are related to our problem when $g = \text{Id}$: a pair of risk measures $(\eta, \rho)$ is called a Bayes pair if there exists a measurable function $G:\Rb^2\rightarrow \Rb$ such that
$$ \eta(Z)=\text{arg}\!\min_{\zeta \in \Rb} \Eb\big[ G(\zeta,Z) \big] \;\text{and}\;\rho(Z)=\min_{\zeta \in \Rb} \Eb\big[ G(\zeta ,Z) \big].$$ In addition, if $\eta$ satisfies that $\eta(Z+c) = \eta(Z) + c$ for any scalar~$c$, then $\rho$ is called a Bayes risk measure. In particular, they show that convex combinations of Expected Shortfall and expectation (called $\text{ES}/\Eb$ mixtures) constitute the unique class of Bayes risk measures that are also coherent risk measures.
Beside $\text{ES}/\Eb$ mixtures, there exist other RB-compatible Bayes risk measures. For instance, the mean absolute deviation around the median
$\text{MAD}(Z):=\min_{\zeta\in \Rb} \Eb[ | Z-\zeta|]$ is Bayes (see \cite{embrechts2021bayes}) and RB-compatible, but not translation invariant (and then not coherent).
More generally, for any $a,b \in \mathbb{R}^*_+$, the risk measure defined by
$$ \rho (Z)= \min_{\zeta \in \Rb} \Eb\Big[a(Z-\zeta)_++ b(Z-\zeta)_{-}\Big]$$
is RB-compatible (see Proposition~\ref{lemma_rm} below) and Bayes. Note that it is not translation invariant, so it is not coherent either.
An interesting and open problem would be to formally characterize RB-compatible Bayes risk measures.

\medskip

In fact, if $\rho$ is an RB-compatible risk measure such that $\mathcal{R}_{\rho, X}(y) =  \min_{\zeta \in \mathcal{Z}} \mathbb E\big[H(\zeta, -y'X) \big]$ for some map $H$, then, for $\beta  \in \mathbb R_+^* $ and $\delta \in \mathbb R$, $\tilde{\rho}$ defined by 
$ \tilde \rho(Z):= \beta \rho(Z) + \delta \Eb[Z]$ is an RB-compatible risk measure too, and 
$\mathcal{R}_{\tilde \rho, X}(y) =  \min_{\zeta \in \mathcal{Z}} \mathbb E\big[\beta H(\zeta, -y'X) - \delta y'X \big]$. In particular, our ideas apply to linear combinations of Expected Shortfall and expectation terms (with positive coefficient for the Expected Shortfall terms). Linear combinations of Expected Shortfall and expectation are in fact quite common in the literature. As we discussed, they appear in the recent literature on Bayes pairs. They also appear for instance in the context of risk measures derived from the notion of optimized certainty equivalent (a utility function-based decision theoretic criterion that was first introduced in \cite{ben1986optimized}): for a large class of utility functions $u$, \cite{ben2007old} introduced sub-additive risk measures of the form 
$$ \rho : Z \mapsto \inf_{\zeta\in \Rb} \Big(\zeta - \Eb\big[u(-Z + \zeta) \big]\Big). $$ 
In particular, if $u(\xi) = a \xi_+ - b \xi_- $ with $0 \le a < 1 < b$, we get 
$$\rho(Z) = \inf_{\zeta \in \mathbb R} \Big( \zeta - \Eb\big[a(-Z+\zeta)_+ - b (-Z+\zeta)_-\big] \Big) = a \mathbb E[Z] + (1-a) \text{ES}_{\frac{b-1}{b-a}}(Z),$$ which is indeed a linear 
combination (with positive coefficients) of an expectation and an Expected Shortfall. 

\mds

Another linear combination of Expected Shortfall and expectation appears when one wants to factor out expectation from Expected Shortfall (think of the initial motivations behind risk-based methods -- see Section \ref{Intro}). 
In that case, we obtain the positive risk measure given by $$\rho(Z) = \text{ES}_{\alpha}(Z) - \mathbb E[Z] = \min_{\zeta \in \mathbb R} \Eb\left[ \zeta - Z + \frac 1{1-\alpha}  (Z-\zeta)_+\right] = \min_{\zeta \in \mathbb R} \Eb\left[ \frac \alpha{1-\alpha} (Z-\zeta)_+ + (Z-\zeta)_-\right].$$ 

\medskip

The above discussion was bound to the specific case where $g = \text{Id}$. Freedom in the choice of $g$ is however important. As an example, if $g(x) = x^2$, $\mathcal{Z} = \mathbb{R}$ and $H(\zeta,Z)= (Z-\zeta)^2$, 
we indeed have, in the case where volatility is the risk measure, that $$g(\mathcal{R}(y)) = \min_{\zeta \in \mathcal{Z}} \mathbb E\big[H(\zeta, -y'X) \big].$$

In other words, the Risk Budgeting problem with volatility as the risk measure can be solved using our new approach.\footnote{Solving the problem by stochastic gradient descent is of course not the best numerical approach in this case.}

\medskip

The natural question is now to evaluate the range of RB-compatible risk measures that can be seen as minima of some expected criteria. As we will see, many RB-compatible risk measures can be written, for well chosen sets $\mathcal{Z}$ and functions $H$, as 
$\min_{\zeta \in \mathcal{Z}} \mathbb E\big[H(\zeta, -y'X)\big]$. Classical generalizations of Expected Shortfall include spectral risk measures: they have such a representation. Moreover, classical extensions of volatility include a large class of deviation measures based on $L^p$ norms, that even includes recent risk measures like variantiles (see \cite{wang2020risk}): we will see they have such a representation too.

\subsection{Spectral risk measures}
Spectral risk measures are well-adopted risk measures that allow to amplify (or reduce) the impact of greater losses through a distortion function, as defined in \cite{acerbi2002spectral}. Formally, they are defined as
\begin{equation*}
\rho_h(Z) := \int_0^1 \text{VaR}_s(Z) h(s) \, ds,
\end{equation*}
for functions $h:[0,1]\rightarrow \Rb_+$ that are non-decreasing, right-continuous and satisfy $\int_0^1 h(s) ds=1$ with $h(0) = 0$, whenever the function $s \mapsto \text{VaR}_s(Z) h(s)$ is integrable on $(0,1)$ (this is guaranteed if $Z \in L^1$ and $h$ is bounded for instance).

\medskip

Expected Shortfall is a particular spectral risk measure since
$$\text{ES}_\alpha(Z) = \frac 1{1-\alpha} \int_\alpha^1 \text{VaR}_s(Z)\, ds = \int_0^1 \text{VaR}_s(Z)   \frac{1}{1-\alpha} 1_{s \in (\alpha,1)} \, ds.$$
In other words, Expected Shortfall corresponds to setting $h(s)=1_{s\in (\alpha,1)}/(1-\alpha)$: it takes into account Value-at-Risks above a fixed threshold uniformly and fully dismisses the others. Spectral risk measures extend Expected Shortfall in that they can assign a variety of weights to the different loss quantiles.

\medskip

For an $L^1$ real-valued random variable $Z$, if $s \mapsto \text{VaR}_s(Z) h(s)$ is integrable on $(0,1)$, then 
\begin{align}
\rho_h(Z) &= \int_0^1 \text{VaR}_s(Z) h(s) \,ds 
= -\int_0^1 \frac{d\ }{ds}\big((1-s) \text{ES}_s(Z)\big)  h(s)\, ds \nonumber \\
&= \int_0^1 (1-s) \text{ES}_s(Z) \, dh(s) - \big[h(s)(1-s)\text{ES}_s(Z) \big]_0^1, \nonumber
\end{align}
where we use Riemann–Stieltjes integrals. As $\lim_{s \to 0^+} h(s)(1-s)\text{ES}_s(Z) = h(0) \Eb[Z] = 0$ and $\lim_{s \to 1^-} h(s)(1-s)\text{ES}_s(Z) = h(1-)  \lim_{s \to 1^-} \int_s^1 \text{VaR}_u(Z)\, du  = 0$, the bracket term vanishes and we get
$$\rho_h(Z) = \int_0^1 (1-s) \text{ES}_s(Z) \, dh(s).$$

In particular, we see that spectral risk measures are part of a large class of RB-compatible risk measures defined by $$\rho^\mu(Z) = \int_0^1 \text{ES}_s(Z) d\mu(s),$$
where $\mu$ is a Borel measure on $(0,1)$.\footnote{This representation has lead to new risk measures called second-order superquantiles. They correspond to $d\mu(s)=1_{s \in (\alpha,1)} \frac 1{1-\alpha} ds$ -- see \cite{rockafellar2018superquantile}.}

\medskip

Interestingly, we have
\begin{align*}
\rho^\mu(Z) &= \int_0^1 \min_{\zeta \in \mathbb R} \Big( \zeta + \frac 1{1-s} \mathbb E\big[(Z-\zeta)_+\big]\Big) \, d\mu(s)\\
&= \min_{\zeta(\cdot) \in \mathcal Z} \int_0^1 \Big( \zeta(s) + \frac 1{1-s} \mathbb E\big[(Z-\zeta(s))_+\big]\Big) \, d\mu(s),
\end{align*}
where $\mathcal Z$ is the set of measurable functions on $(0,1)$.

\medskip

When $\mu(A) = \int_{A\cap(0,1)} (1-s) dh(s)$ for any Borel set $A$, we get
\begin{align}
\rho_h(Z) &= \min_{\zeta(\cdot) \in \mathcal Z} \mathbb E\bigg[\int_0^1 \Big((1-s)\zeta(s) + \big(Z-\zeta(s)\big)_+\Big) \, dh(s)\bigg]. \label{express_spectral}
\end{align}

\medskip
\begin{rem}
As in the case of Expected Shortfall, it is possible to factor out expected returns from spectral risk measures. Noticing that $\int_0^1 (1-s) dh(s) = 1$, we obtain the positive risk measure
\begin{align*}
\rho_h(Z) - \mathbb E[Z] &= \min_{\zeta(\cdot) \in \mathcal Z} \mathbb E\left[\int_0^1 \Big( (1-s)\big(\zeta(s) - Z\big) + \big(Z-\zeta(s)\big)_+\Big) \, dh(s)\right]   \\
&=\min_{\zeta(\cdot) \in \mathcal Z} \mathbb E\left[\int_0^1 \Big( s \big(Z-\zeta(s)\big)_+ + (1-s) \big(Z-\zeta(s)\big)_-\Big) \, dh(s)\right].
\end{align*}
\end{rem}

In the definition of Bayes risk measure, the set $\Zc$ is the real line. Here, we noticed that solving Risk Budgeting problems for spectral risk measures boils down to solving (stochastic) optimization problems, but in infinite dimension because $\Zc$ is the set of measurable functions. Thus, it is tempting to say that a spectral risk measure is ``weakly'' Bayes. Note that the only coherent Bayes risk measures are $\text{ES}/\Eb$ mixtures whereas any spectral risk measure is
coherent.

\medskip

In practice, we need to discretize the integrals to get a finite-dimensional stochastic optimization problem. These approximate solutions correspond to risk measures that are linear combinations (with positive coefficients) of a finite number of Expected Shortfall terms for different risk levels.

\subsection{Deviation measures}
Deviation measures are ubiquitous in statistics and applied mathematics. A risk measure $\rho$ is a deviation measure if it satisfies \textbf{(PH)},~\textbf{(SA)}, $\rho(Z+c)=\rho(Z)$ for any random variable $Z$ such that $\rho(Z)$ is defined and any constant $c$, and $\rho(Z)>0$ when $Z$ is not a constant. In other words, a deviation measure is RB-compatible, positive, and left unchanged when a riskless asset is added to any portfolio. See \cite{rockafellar2006generalized,rockafellar2008risk,rockafellar2013fundamental} for a discussion on deviation measures in risk management particularly. 
Standard deviation is surely the most popular member of this class, and it is commonly used in finance. 
Indeed, being the square root of a quadratic form, standard deviation is appealing for introducing the risk dimension in portfolio optimization problems.
It is however symmetrical because the volatility of $X$ is the same as the volatility of $-X$, while there is a great benefit in considering gains and losses asymmetrically in finance. 
In what follows, we propose a large class of RB-compatible risk measures that contains both symmetrical and asymmetrical deviation measures and for which our ideas of stochastic optimization (see Section~\ref{tsop}) apply.

\medskip

\begin{prop}
\label{lemma_rm}
Let $a,b \in \mathbb{R}^*_+$ and let the function $\psi_{a,b}:\Rb\rightarrow \Rb_+$ be defined by 
$$ \psi_{a,b}: z \longmapsto a z_+ + b z_-.$$
Set $p \in [1, +\infty)$. Let $F$ be a finite-dimensional subspace of $L^p(\Omega,\mathbb{R})$ and $\rho : L^p(\Omega,\mathbb{R}) \longrightarrow \mathbb{R}$ be defined by
$$ \rho : Z \longmapsto \inf_{f \in F} \Eb\big[{\psi_{a,b}(Z-f)^p}\big]^{\frac{1}{p}}.$$
Then, $\rho$ is an RB-compatible risk measure and the infimum in the definition of $\rho$ is in fact a minimum. 
\end{prop}

\begin{proof} 
For $\lambda>0$, we have $$\rho(\lambda Z) = \inf_{f \in F} \Eb\big[{\psi_{a,b}(\lambda Z-f)^p}\big]^{\frac{1}{p}}=\inf_{f \in F} \Eb\left[\lambda^p{\psi_{a,b}\left( Z-\frac{f}{\lambda}\right)^p}\right]^{\frac{1}{p}}$$$$= \lambda \inf_{f \in F} \Eb\left[{\psi_{a,b}\left( Z-\frac{f}{\lambda}\right)^p}\right]^{\frac{1}{p}} = \lambda\inf_{f \in F} \Eb\big[{\psi_{a,b}( Z-f)^p}\big]^{\frac{1}{p}} = \lambda \rho(Z).$$
Since $\rho(0) = \inf_{f \in F} \Eb[{\psi_{a,b}(-f)^p}]^{\frac{1}{p}}=0$, $\rho$ is positively homogeneous.

\medskip

Coming to sub-additivity, it is clear that $\psi_{a,b}$ is sub-additive. Thus, for all $Z_1,Z_2  \in  L^p(\Omega,\mathbb{R})$, we have

\begin{align*}
\rho(Z_1+Z_2) = \inf_{f \in F} \Eb\big[{\psi_{a,b}(Z_1+Z_2-f)^p}\big]^{\frac{1}{p}} &= \inf_{f_1,f_2 \in F} \Eb\big[{\psi_{a,b}(Z_1+Z_2-f_1-f_2)^p}\big]^{\frac{1}{p}} \\
&\leq \inf_{f_1,f_2 \in F} \Eb\big[{(\psi_{a,b}(Z_1-f_1)+\psi_{a,b}(Z_2-f_2))^p}\big]^{\frac{1}{p}} \\
& \leq \inf_{f_1,f_2 \in F} \Eb\big[{(\psi_{a,b}(Z_1-f_1))^p}\big]^{\frac{1}{p}} + \Eb\big[{(\psi_{a,b}(Z_2-f_2))^p}\big]^{\frac{1}{p}}\\ 
& \leq \rho(Z_1)+\rho(Z_2),
\end{align*}
where we used the triangular inequality for the $L^p$ norm.
\medskip

Let us now consider a sequence $(f_n)_n$ of maps in $F$ such that $\Eb\big[{\psi_{a,b}( Z-f_n)^p}\big]^{\frac{1}{p}} \le \rho(Z) + \frac 1{n+1}$. We have
$$ \Eb[f_n^p]^{\frac{1}{p}} \le  \Eb[Z^p]^{\frac{1}{p}} +  \Eb[|Z-f_n|^p]^{\frac{1}{p}} \le \Eb[Z^p]^{\frac{1}{p}}+ \frac{1}{\min(a,b)} \Eb\big[\psi_{a,b}(Z-f_n)^p\big]^{\frac{1}{p}} $$$$\le \Eb[Z^p]^{\frac{1}{p}} + \frac{1}{\min(a,b)}\big(\rho(Z) +1\big)$$ and therefore $(f_n)_n$ is bounded in $F$. Up to a subsequence, it converges therefore, for the $L^p$ norm, towards a random variable $f \in F$ and we have

\begin{align*}
\Eb\big[\psi_{a,b}(Z-f)^p\big]^{\frac{1}{p}} &\le \Eb\Big[\big(\psi_{a,b}(Z-f_n) + \max(a,b) |f_n-f|\big)^p\Big]^{\frac{1}{p}} \\ 
&\le \Eb\big[\psi_{a,b}(Z-f_n)^p\big]^{\frac{1}{p}} + 
\max(a,b)  \Eb\big[|f_n-f|^p]^{\frac{1}{p}}\big] \\
&\le \rho(Z) + \frac 1{n+1} + \max(a,b)  \Eb[|f_n-f|^p]^{\frac{1}{p}}.
\end{align*}
Sending $n$ to $+\infty$, we see that $f$ is a minimizer and that the infimum in the definition of $\rho$ is indeed a minimum.
\end{proof}

Since $ \rho(Z+c)=\rho(Z)$ for all $Z \in L^p(\Omega,\mathbb{R})$ and any constant $c$, the latter risk measures $\rho$ are called deviation measures -- see~\cite{rockafellar2002deviation}. They are compatible with our framework for $g(x) = x^p$. Interestingly, the above family contains some familiar risk measures for particular choices of $F$, $a$, $b$ and $p$:

\medskip

(a) when $a=b$, we get symmetrical measures as
\begin{itemize}
    \item standard deviation when $F = \text{span}(1)$, $a=b=1$ and $p=2$;
    \item mean absolute deviation around the median (MAD), i.e. $\Eb\big[ |Z - \text{median}(Z)|\big]$, when $F = \text{span}(1)$, $a=b=1$ and $p=1$ (the minimum is reached for the median of the portfolio losses).
\end{itemize}

(b) When $a \neq b$, we retrieve some asymmetrical measures, for instance
\begin{itemize}
    \item Expected Shortfall at level $\alpha$ minus expectation, i.e. $\text{ES}_\alpha(Z) - \mathbb{E}[Z]$, for $\alpha \in (0,1)$, when $F = \text{span}(1)$, $a=\alpha/(1-\alpha)$, $b=1$ and $p=1$;
    \item the square root of the variantile at level $\alpha$ (see \cite{wang2020risk}) when $F = \text{span}(1)$, $a=\sqrt{\alpha}$, $b=\sqrt{1-\alpha}$ and $p=2$ (the minimum is reached for the expectile at level $\alpha$).
\end{itemize}

The above examples show the relevance of this family of risk measures, which are particular cases of deviation measures. 
When the goal is to focus on heavy tail risks, it should be relevant to impose $p>2$.
Extensions beyond the space $F$ of constant random variables can also be considered to focus on residual risk (when $F$ is spanned by factors, in the same spirit as~\cite{rockafellar2008risk}).

\section{Numerical results}
\label{Numerical_results}
\subsection{Convergence of the stochastic gradient descent algorithm}
In this section, we want to illustrate our results and ideas in a simple case. We consider a set of~4 assets and construct an ERC portfolio\footnote{This corresponds to setting risk budgets to $b = (0.25,0.25,0.25,0.25)$.} for Expected Shortfall at a confidence level $\alpha= 95\%$. Our goal is to compare the performances of an SGD algorithm with that of a standard procedure.\footnote{Although we call our method an SGD method, we use a mini-batch implementation because of its computational advantage compared to vanilla stochastic gradient descent where one considers a single data point at each iteration.}

\medskip
For what follows, we assume that the joint distribution of our asset returns is given by a mixture of two multivariate Student-t distributions. More precisely, we assume that $X$ has the following density with respect to the Lebesgue measure:
\begin{equation*}
    f_X(x) := p f(x|\mu_1, \Lambda_1, \nu_1) + (1-p) f(x|\mu_2, \Lambda_2, \nu_2).
\label{t_density}
\end{equation*}
To work with a realistic model, it has been calibrated using daily returns of Apple Inc. (AAPL), JPMorgan Chase \& Co. (JPM), Pfizer Inc. (PFE) and Exxon Mobil Corporation (XOM) over the period August 2008--April 2022. 
Our model parameters are then estimated using the expectation-maximization algorithm and rounded for the purpose of illustration.
Fixing the degrees of freedom at $\nu_1=4.0$ and $\nu_2=2.5$ a priori, we obtained that the weight is $p=0.7$, the location vectors are $\mu_1 = ( 0.001,  0.001,  0.001,  0.003)'$ and $\mu_2 = (-0.001, -0.002, -0.001, -0.002)'$, and the scale matrices are
\begin{center}
$\Lambda_1\! =\! \begin{bmatrix}
0.00010 &  0.00005 &  0.00002 &  0.00003 \\
0.00005 &  0.00010 &  0.00002 &  0.00002 \\
0.00002 &  0.00002 &  0.00010 &  0.00002 \\
0.00003 &  0.00002 &  0.00002 &  0.00010 \\
\end{bmatrix}$ and 
$\Lambda_2\! =\! \begin{bmatrix}
0.00040 &  0.00010 &  0.00010 &  0.00020 \\
0.00010 &  0.00010 &  0.00008 &  0.00009 \\
0.00010 &  0.00008 &  0.00010 &  0.00007 \\
0.00020 &  0.00009 &  0.00007 &  0.00020 \\
\end{bmatrix}$.
\end{center}

\medskip

Since Expected Shortfall has a semi-analytic form in our multivariate Student-t mixture model (Section~\ref{parametric_model_analytic}), we can easily compute the Risk Budgeting portfolio using the L-BFGS-B algorithm.\footnote{We stop the algorithm when the infinity norm of a projected gradient vector is less than $10^{-6}$.} The resulting  portfolio~$\theta$ is given in Table~\ref{tab-optimal-portfolio}. We confirm that $\theta$ solves $\text{RB}_{b}$ by noticing that the risk contributions $\theta_i \partial_i \mathcal{R}(\theta)$ are the same for all assets. This portfolio constitutes a reliable reference to evaluate the convergence and the accuracy of alternative (stochastic) methods. It is referred to as {\it the reference portfolio} in what follows.

\medskip

\begin{center}
\begin{tabular}{ccc}
\toprule
{Asset} &   $\theta_i$ &  $\theta_i \partial_i \mathcal{R}(\theta)$\\
\midrule
1 &  0.17958 &     0.00806 \\
2 &  0.28127 &     0.00806 \\
3 &  0.30483 &     0.00806 \\
4 &  0.23432 &     0.00806 \\
\bottomrule
\end{tabular}
\captionof{table}{Reference portfolio weights and risk contributions.}
\label{tab-optimal-portfolio}
\end{center}

Let us now come to the use of SGD methods to compute Risk Budgeting portfolios. In order to solve the stochastic optimization problem presented in Section~\ref{tsop} using SGD, we require sample points for asset returns. In what follows, we draw $10^6$ sample points with the multivariate Student-t mixture distribution of the random variable $X$. We use a mini-batch size of 128 and 10 epochs. Figure~\ref{sgd_conv} illustrates the dynamics of the $(y,\zeta)$ pair in~(\ref{stocRB}) and that of the associated $\theta^{SGD}$ through the iterations.

\medskip

\begin{figure}[!ht]
\centering
     \includegraphics[width=1.0\textwidth]{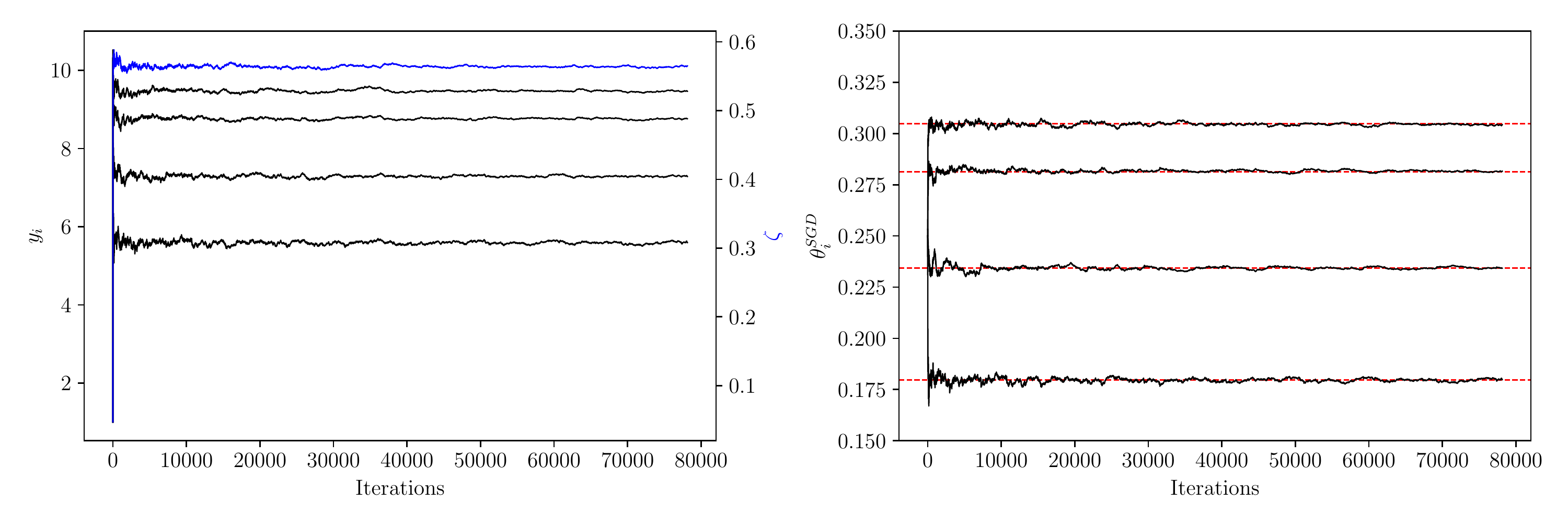}
      \caption{Left: evolution of the components of $y$ (black / left axis) and $\zeta$ (blue / right axis). Right: evolution of $\theta^{SGD}$ -- dashed lines are the asset weights of the reference portfolio.}
       \label{sgd_conv}
\end{figure}
\vspace{3mm}

The final estimator $y^{*^{SGD}}$ is computed using a standard Polyak-Ruppert averaging that relies on the last 20\% of all iterations. We always use this variance reduction technique for the SGD method throughout the paper. The resulting portfolio $\theta^{*^{SGD}}$ and its deviation from the reference portfolio are given in Table~\ref{tab-optimal-sgd_ns}. We clearly see that $\theta^{*^{SGD}}$ is very close to the reference portfolio.
\medskip

\begin{center}
\begin{tabular}{ccc}
\toprule
Asset  &  $\theta^{*^{SGD}}_i$ &      $|\theta_i-\theta^{*^{SGD}}_i|$ \\
\midrule
 1 &  0.17954  &     0.00005  \\
 2 &  0.28165 & 0.00038  \\
 3 &  0.30449  &  0.00034  \\
 4 &  0.23432  &    0.00001  \\
\bottomrule
\end{tabular}
\captionof{table}{Risk Budgeting portfolio weights obtained by using the SGD method with a Polyak-Ruppert averaging.}
\label{tab-optimal-sgd_ns}
\end{center}

\subsection{Speed and accuracy of various methods for different portfolio sizes}
\label{comparison}
Semi-analytic formulas for Expected Shortfall are only available in very specific cases. In general, the approximated computation of an Expected Shortfall without using the Rockafellar-Uryasev formula requires a sample of returns $\mathcal{X} = \{x_1, \ldots, x_{n}\}$. Such returns are observed historically or may be simulated based on a model. Then, invoke the usual empirical estimator
$$ \widehat{\text{ES}}_\alpha(y, \mathcal{X})= \frac{\sum^{n}_{i=1}(-y'x_i) 1_{\{-y'x_i \ge \hat q_\alpha\}}}{\sum^{n}_{i=1}1_{\{-y'x_i \ge \hat q_\alpha\}}},$$  where $\hat q_\alpha$ is the empirical Value-at-Risk of level $\alpha$, i.e. it is the empirical $\alpha-$quantile\footnote{There exist several ways of computing empirical quantiles. We shall use the default method of Python NumPy package which corresponds to the method 7 in \cite{hyndman1996sample}.} associated with the set of portfolio losses $\{-y'x_1, \ldots, -y'x_{n}\}$.

\medskip
In this case, the computation of the Risk Budgeting portfolio corresponds to solving the optimization problem given by
\begin{equation}
\label{empRB}
    \min_{y \in (\mathbb R_+^*)^d}  \big( \widehat{\text{ES}}_\alpha(y, \mathcal{X}) - \sum_{i=1}^d b_i \log{y_i} \big),
\end{equation}
and normalizing the minimizer.

\medskip
To solve Problem~\eqref{empRB}, one can use a standard gradient descent approach where, at each iteration, the gradient is approximated using finite differences, i.e. $ \partial_{y_i} \text{ES}_\alpha(-y'X)$ is approximated by
\begin{equation*}
\frac{\widehat{\text{ES}}_\alpha(y+he_i, \mathcal{X}) - \widehat{\text{ES}}_\alpha(y, \mathcal{X})}{h},
\end{equation*}
where $e_i$ denotes the $i^{\text{th}}$ vector of the canonical basis of $\mathbb R^d$ and $h$ is a small tuning parameter. Hereafter, we set $h=10^{-4}$.

\medskip

Two methods can be derived from the above remarks to compute Risk Budgeting portfolios for Expected Shortfall.
The first one corresponds to using the same sample of asset returns throughout the optimization process to estimate these risk measures. It is the only option (besides the SGD method) when working with historical samples in a model-free setting. This method will be called the \textit{one-sample benchmark gradient descent} (OSBGD) approach.  Such ideas can also be used in a model-based setting where one can run the algorithm on a simulated sample. This alternative way of working requires a model for $X$ and is based on simulating a new sample of asset returns at each iteration to estimate the new Expected Shortfall. We call the latter method the \textit{multi-sample benchmark gradient descent} (MSBGD) approach. The advantage of this second method is that it allows to manage smaller sample sizes and thus faster calculations at each iteration, while allowing for comparable -- or in fact higher -- accuracy in the long run. 

\medskip

Our implementation of the OSBGD method is a gradient descent procedure which is based on the Barzilai-Borwein methodology: at each step, the gradient is calculated using finite differences over a fixed sample and the step size is determined by a cheap approximation of the Hessian (see~\cite{barzilai1988two}). The stopping rule is based on the difference between two consecutive values of the objective function in Problem~\eqref{empRB} and the algorithm is stopped if that difference drops below~$10^{-6}$. The MSBGD method is very similar to the above procedure, but the gradient is calculated using a new sample drawn from the chosen model of asset returns at each iteration. The algorithm is stopped after a fixed number of iterations rather than using a stopping rule because of the stochasticity of gradient approximations. The final estimator is computed by averaging the last iterations.  

\medskip
This section provides results about accuracy and speed of the SGD, OSBGD and MSBGD methods in model-free and model-based settings. Our accuracy measure is the Manhattan / $\| \cdot \|_1 $ distance between the reference portfolio and the portfolio obtained by the optimization method.

\medskip
We are interested in building ERC portfolios of $d \in \{10, 20, 50, 100, 200, 350\}$ assets for Expected Shortfall at the confidence level $\alpha= 95\%$. In practice, as we are never in possession of neither the true distribution of asset returns nor the reference portfolio, we cannot measure the accuracy of the portfolio computed by the presented methods. To conduct our empirical analysis, we adopt a simulation-based approach. We define a data generating process (DGP) and assume that it reflects the true distribution of asset returns. We call it $\text{DGP}_{\text{true}}$. We then draw $n=3500$ data points from $\text{DGP}_{\text{true}}$ to generate a \textit{synthetic} historical sample $\mathcal{X}_{hist}$.\footnote{The sample size $n$ is chosen so as to represent the typical size of historical samples in the equity world where we often deal with a maximum of 10--15 years of daily return data.} The use of $\mathcal{X}_{hist}$ to compute Risk Budgeting portfolios is twofold: we can follow the model-free approach, where we can run the SGD and OSBGD methods using $\mathcal{X}_{hist}$, or, alternatively, the model-based approach, where we can fit a model on $\mathcal{X}_{hist}$ and proceed with simulated samples $\mathcal{X}_{sim}$ drawn from the estimated model. The latter choice allows to use all three methods (SGD, OSBGD and MSBGD). 
A detailed description of all these procedures is given in Appendix~\ref{appendixA}.

\medskip
Start with the model-free approach and compute Risk Budgeting portfolios using the SGD and OSBGD methods using $\mathcal{X}_{hist}$. For the SGD method, we use a mini-batch size of 128 and stop it after 100 epochs. Table~\ref{practical_both} documents the accuracy and computation time of both methods for different portfolio sizes.

\medskip

Table~\ref{practical_both} shows that one can get reasonably close to the \textit{true} Risk Budgeting portfolio up to a certain level with a limited amount of historical data. We observe that both methods yield very similar results in
terms of accuracy. It is intuitive that both methods produce similar results after an almost complete process of the information contained in the same inputs. In terms of computation time, the OSBGD method is efficient especially when constructing portfolios with a small number of assets. The advantage of OSBGD over SGD however disappears as the portfolio dimension $d$ grows. 

\medskip

\begin{center}
\begin{tabular}{c|cc|cc}
\toprule
{} & \multicolumn{2}{c|}{Accuracy} & \multicolumn{2}{c}{Time} \\
$d$ &          SGD &        OSBGD &          SGD &        OSBGD \\
\midrule
10  &  5.46 (1.63) &  5.47 (1.65) &  1.08 (0.01) &  0.05 (0.02) \\
20  &  6.63 (1.82) &  6.63 (1.78) &  1.18 (0.01) &  0.14 (0.09) \\
50  &  7.26 (1.64) &  7.28 (1.68) &  1.32 (0.01) &  0.34 (0.12) \\
100 &  7.67 (1.06) &  7.69 (1.05) &  1.52 (0.01) &  0.76 (0.30) \\
200 &  7.60 (1.42) &  7.60 (1.41) &  1.91 (0.02) &  1.32 (0.43) \\
350 &  7.83 (1.73) &  7.73 (1.62) &  2.52 (0.01) &  2.53 (1.09) \\
\bottomrule
\end{tabular}
 \captionof{table}{Accuracy of the Risk Budgeting portfolios obtained by the SGD and OSBGD methods for different numbers of assets under historical samples and computation time of algorithms (in seconds). The accuracy measure corresponds to $100 \lVert \theta - \theta^{method} \rVert_1$. Figures correspond to means and standard deviations (in parentheses) computed by repeating the process $m=50$ times with $\mathcal{X}_{hist}$ drawn from $m$ different $\text{DGP}_{\text{true}}$ for each $d$.}
\label{practical_both}
\end{center}

The alternative to the model-free approach is to follow the model-based approach. It uses $\mathcal{X}_{hist}$ to evaluate a model that is believed to reflect the true behavior of asset returns. Then, such a model allows us to draw large simulated samples without being restricted by the size of the historical sample. The primary risk of this approach is the mis-specification of the true distribution of asset returns. We therefore want to consider three cases. The first one corresponds to a situation where the estimated model perfectly matches the true distribution of asset returns -- $\text{DGP}_{\text{true}}$. This case is not realistic but worth to analyze. Of course, we expect to obtain results very close to the reference portfolios. 
In the second case, we correctly specify the family of the true distribution of $X$: we assume that $\text{DGP}_{\text{true}}$ is really a mixture of two multivariate Student-t distributions. Naturally, in our case, well-specifying the parametric family of the $X$ distribution corresponds to the case where we fit a mixture of two multivariate Student-t distributions to $\mathcal{X}_{hist}$. This yields an estimated model $\text{DGP}_{\text{SM}}$.\footnote{We use the expectation-maximization algorithm with fixed degrees of freedom: $\nu_1=4.0$ and $\nu_2=2.5$.} For the third method, we mis-specify the family of the $X$ distribution. In our case, we fit a mixture of two multivariate Gaussian distributions to $\mathcal{X}_{hist}$ and obtain a ``wrong'' model $\text{DGP}_{\text{GM}}$. 
Obviously, any other parametric family can be assumed. The mis-specified case is here a Gaussian mixture so as not to excessively deviate from $\text{DGP}_{\text{true}}$.

\medskip
Table~\ref{simulation_methods} shows the results of the model-based approach. For the SGD and OSBGD methods, we run the algorithm using a (fixed) simulated sample $\mathcal{X}_{sim}$ of size $10^6$. We use a mini-batch size of 128 and stop it after 4 epochs for the SGD method. For the MSBGD method, the size of the sample $\mathcal{X}_{sim}$ -- that is repeatedly drawn over again at each iteration -- is chosen to be $10^5$ and we stop the algorithm after 60 iterations. The final estimator is computed by averaging the last 5 iterations.\footnote{No substantial improvement is observed after 4 epochs and 60 iterations for the SGD and MSBGD methods respectively. The sample size used in the MSBGD method is chosen to be $10^5$ because, using smaller sample sizes like $10^3$ and $10^4$, we cannot get accurate portfolios due to very poor quality of gradient approximation. On the other hand, a larger sample size like $10^6$ excessively slows down the process for a negligible improvement in terms of accuracy.}

\medskip
\begin{center}
\begin{adjustbox}{width=\textwidth}
\begin{tabular}{c|c|ccc|ccc}
\toprule
                       &     & \multicolumn{3}{c|}{Accuracy} & \multicolumn{3}{c}{Time} \\
                       &  $d$    &          SGD &         OSBGD &         MSBGD &           SGD &            OSBGD &          MSBGD \\
\midrule
Well-specified & 10  &  0.37 (0.10) &  0.35 (0.10) &  0.34 (0.08)  &   3.42 (0.12) &    12.82 (2.43) &   9.64 (0.35) \\  
(Parameters)               & 20  &  0.42 (0.09) &  0.41 (0.12) &  0.37 (0.12)  &   3.75 (0.06) &     26.00 (4.64) &   16.49 (0.89) \\
           & 50  &  0.52 (0.12) &  0.49 (0.12) &  0.38 (0.10)  &   4.46 (0.09) &    66.10 (13.01) &   32.77 (0.52) \\
               & 100 &  0.51 (0.04) &  0.49 (0.05) &  0.40 (0.12)  &   5.58 (0.06) &    141.21 (23.62) &   62.27 (1.19) \\ 
               & 200 &  0.53 (0.08) &  0.53 (0.10) &  0.40 (0.05)  &   9.05 (0.62) &  275.14 (52.35) &   123.96 (1.35) \\
               & 350 &  0.62 (0.09) &  0.54 (0.09) &  0.41 (0.08)  &   17.51 (0.76) &  522.43 (64.11) &   221.69 (1.04) \\

\midrule
Well-specified & 10  &   2.91 (1.30) &   2.91 (1.31) &    3.00 (1.36) &    3.69 (0.30) &     11.97 (2.43) &    9.18 (0.51) \\
 (Family)                      & 20  &  3.22 (2.03) &    3.20 (2.02) &   3.23 (1.87) &   3.81 (0.12) &     24.97 (5.54) &   16.79 (0.59) \\
                       & 50  &  3.91 (1.63) &   3.93 (1.64) &   3.93 (1.74) &   4.41 (0.05) &    67.41 (11.98) &    32.00 (1.32) \\
                       & 100 &  3.25 (1.36) &   3.24 (1.36) &   3.23 (1.37) &   5.52 (0.04) &   130.35 (14.46) &   60.73 (1.34) \\
                       & 200 &   3.70 (1.62) &   3.65 (1.64) &   3.61 (1.57) &    8.77 (0.90) &   278.52 (44.32) &  118.22 (1.19) \\
                       & 350 &  4.25 (0.81) &   4.03 (0.75) &   3.97 (0.75) &  16.41 (0.52) &   458.99 (53.37) &  213.48 (1.02) \\
\midrule
Mis-specified & 10  &  8.38 (3.01) &   8.39 (3.01) &   8.45 (2.99) &   3.63 (0.29) &     14.66 (3.08) &    9.08 (0.28) \\
                       & 20  &  8.01 (2.55) &   8.01 (2.54) &   8.01 (2.53) &   3.89 (0.17) &     29.98 (6.14) &   16.89 (0.68) \\
                       & 50  &  8.76 (2.75) &   8.77 (2.79) &    8.70 (2.78) &   4.44 (0.04) &   107.07 (80.33) &   32.53 (2.14) \\
                       & 100 &  8.55 (5.06) &   8.54 (5.07) &   8.49 (5.01) &    5.53 (0.10) &   156.53 (58.07) &   60.57 (1.48) \\
                       & 200 &  9.23 (3.94) &   9.22 (3.94) &   9.11 (3.74) &   8.85 (0.68) &   337.73 (94.27) &   118.50 (2.16) \\
                       & 350 &  10.45 (3.20) &   10.00 (2.99) &   10.01 (3.04) &  16.72 (0.65) &  637.16 (153.91) &  213.95 (1.08) \\
\bottomrule
\end{tabular}
\end{adjustbox}
 \captionof{table}{Accuracy of the Risk Budgeting portfolios obtained by the three methods for different numbers of assets for samples drawn from $\text{DGP}_{\text{True}}$, $\text{DGP}_{\text{SM}}$ and $\text{DGP}_{\text{GM}}$ and computation time of algorithms (in seconds). The accuracy measure corresponds to $100 \lVert \theta - \theta^{method} \rVert_1$. Figures correspond to means and standard deviations (in parentheses) computed by repeating the process $m=50$ times with samples drawn from $m$ different $\text{DGP}_{\text{True}}$, $\text{DGP}_{\text{SM}}$ and $\text{DGP}_{\text{GM}}$ for each $d$.}
\label{simulation_methods}
\end{center}

In Table~\ref{simulation_methods}, we observe that the model-based approach can yield more accurate results than the model-free approach whose results are shown in Table~\ref{practical_both}. It is of course true in the case of a perfect specification of the distribution of asset returns. Interestingly, this is still true in the case of a correct specification of the family of the $X$ distribution. However, we see that it generally does not hold any longer in the case of a mis-specified model. We can conclude that model-based methods can challenge model-free methods if we are confident about our choice of the family of the true distribution of asset returns. Another important result concerns the computation time of methods. The MSBGD method is faster than the OSBGD method since the gradient can be computed much faster when $\mathcal{X}_{sim}$ is of size $10^5$ compared to the case of OSBGD where it is of size~$10^6$. However, its performance in terms of speed is also far from being close to the performance of the SGD method especially for high-dimensional portfolios. The results show that SGD is the fastest among the three methods to obtain accurate portfolios independently of the number of assets if we adopt the model-based approach.

\subsection{Risk Budgeting for the allocation of negatively skewed assets}
Choosing the risk measure that correctly reflects the \textit{true} risk of an investment is an important task in portfolio construction. In this section, we build Risk Budgeting portfolios using the SGD method for the risk measures mentioned in Section~\ref{Extending_framework}. Our aim is to illustrate the impact of the choice of a risk measure on the Risk Budgeting portfolio using a simple example. We will consider a case in which asset returns are normally distributed and a second case where asset returns exhibit jump risk and are negatively skewed.
Indeed, negative skewness of asset returns distributions is generally associated with large negative jumps in asset prices. A reasonable modeling approach (see~\cite{bruder2016risk} for a similar model) is to assume a two-component mixture model. Each component represents a different state of the market, typically a ``normal'' state and a ``stressed'' state where the probability of a downward jump substantially increases.

\medskip
Here, consider three assets and a mixture of two multivariate Gaussian distributions: the joint density of their returns $X$ 
is given by 
\begin{equation*}
    f_X(x) = p \phi(x|\mu_1, \Sigma_1) + (1-p) \phi(x|\mu_2, \Sigma_2),
\label{g_density}
\end{equation*}
where $\phi(\mu,\Sigma)$ denotes the probability density function of a multivariate Gaussian distribution $\Nc(\mu,\Sigma)$ and $p$ is the probability of being in the ``normal'' market state. We consider the following parameters, ensuring that the parameters of the ``stressed'' state introduce negative skewness to some asset returns. The vector of expected returns in the normal market state is  $\mu_1 = (0.02, 0.06, 0.10)'$. In the stressed market state, the expected returns of some assets dramatically decrease: $\mu_2 = (-0.15, -0.30,  0.10)'$. The covariance matrices in the two different states are  

\begin{center}
$\Sigma_1\! =\! \begin{bmatrix}
0.0064 & 0.0080 & 0.0048 \\
0.0080 & 0.0400 & 0.0240 \\
0.0048 & 0.0240 & 0.0900 \\
\end{bmatrix}$ and
$\Sigma_2\! =\! \begin{bmatrix}
0.0289 & 0.0230 & 0.0048 \\
0.0230 & 0.0800 & 0.0240 \\
0.0048 & 0.0240 & 0.1000 \\
\end{bmatrix}$.
\end{center}

The marginal distributions of asset returns under these parameters are illustrated in Figure~\ref{asset_dists}.

\vspace{0mm}
\begin{figure}[!ht]
\hspace{-1.6cm}
     \includegraphics[width=1.2\textwidth]{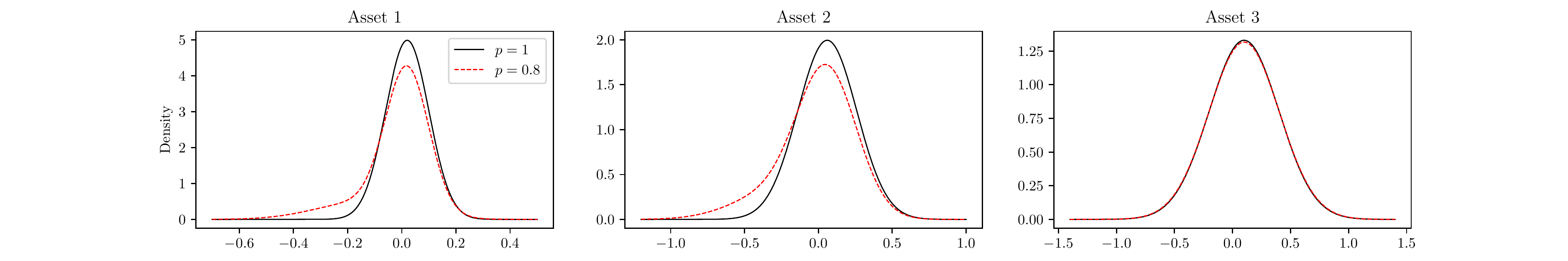}
      \caption{Marginal distributions of asset returns based on the chosen parameters when $p=1$ and $p=0.8$.}
       \label{asset_dists}
\end{figure}
\vspace{0mm}

We construct ERC portfolios with these three assets and several risk measures: volatility, mean absolute deviation around the median $\text{MAD}(Z) = \min_{\zeta \in \mathbb R} \mathbb E\big[|Z-\zeta|\big]$, Expected Shortfall $\text{ES}_{\alpha}$, spectral risk measure $\rho_{h}$ where the distortion function $h$ is a power function and variantile $\upsilon_{\alpha}(Z) = \min_{\zeta \in \mathbb R} \mathbb E\big[\alpha (Z - \zeta)_{+}^2 + (1-\alpha) \mathbb  (Z - \zeta)^2_{-} \big] $. We also examine the impact of adding or factoring out expected loss when relevant.

\medskip
To be more explicit, $\rho_{h}$ is based on the distortion function $h(s) = s^{1/c-1}/c$ for some $c\in(0,1]$. As $c$ gets close to zero, $h(s)$ attributes more weight to larger $s$ values, making the risk measure more sensitive to extreme losses. In this section, we consider the case $c=0.05$. This function $h$ might reflect the risk profile of an investor better than a step function -- as in the case of Expected Shortfall -- because it assigns increasing weights to larger losses in a smooth way.

\medskip
Table~\ref{risk_measures_combined} shows the estimated RB portfolios for $p \in \{0.8,1\}$.

\medskip

\begin{center}
\begin{tabular}{c|c|ccc}
\toprule
{} & risk measure & Asset 1 &  Asset 2 &  Asset 3 \\
\midrule
$p=1$   &   Volatility   &    0.60916 &  0.22200 &  0.16884 \\
   &   MAD          &    0.60951 &  0.22185 &  0.16864 \\
   &   $\text{ES}_{0.95} - \Eb$     &    0.60972 &  0.22190 &  0.16838 \\  
   &   $\rho_{h} - \Eb$   &        0.60969 &  0.22200 &  0.16831 \\
      &   $\text{MAD} + \Eb$      &    0.58872 &  0.22046 &  0.19082\\
   &   $\text{ES}_{0.95}$        &    0.60342 &  0.22168 &  0.17490 \\
   &   $\rho_{h}$   &        0.60252 &  0.22169 &  0.17579 \\
   &   $\upsilon_{0.99}$ &    0.59850 &  0.22138 &  0.18012 \\
\toprule
$p=0.8$   &   Volatility   &    0.52700 &  0.22882 &  0.24418 \\
   &   MAD          &    0.54790 &  0.22644 &  0.22566 \\
      &   $\text{ES}_{0.95} - \Eb$     &    0.46458 &  0.22612 &  0.30929 \\
   &   $\rho_{h} - \Eb$   &        0.47528 &  0.22727 &  0.29745 \\
   &   $\text{MAD} + \Eb$      &    0.45476 &  0.20345 &  0.34180\\
   &   $\text{ES}_{0.95}$        &    0.44055 &  0.21511 &  0.34434 \\

   &   $\rho_{h}$   &        0.44515 &  0.21510 &  0.33975 \\

   &   $\upsilon_{0.99}$ &    0.45719 &  0.21327 &  0.32954 \\
\bottomrule
\end{tabular}
\captionof{table}{Risk Budgeting portfolios for different risk measures under the assumption of normal ($p=1$) and negatively skewed (Gaussian mixture) asset returns ($p=0.8$).}
\label{risk_measures_combined}
\end{center} 

\mds

When $p=1$, asset returns follow a multivariate Gaussian distribution and do not exhibit skewness. We obtain very similar Risk Budgeting portfolios for all the risk measures insensitive to expected loss, i.e. volatility, MAD, $\text{ES}_{0.95} - \Eb$ and $\rho_{h} - \Eb$, because these risk measures are proportional to one another in the Gaussian case. When expected loss comes into play, it slightly impacts the allocation. Overall, in this case, there is no apparent advantage in using a risk measure different from volatility in the absence of negative skewness.

\medskip

When we introduce skewness by setting $p=0.8$, we obtain significantly different Risk Budgeting portfolios. Volatility and MAD seem to capture part of the higher risk induced by the likelihood of observing a stressed market. However, symmetrical deviation measures are not ideal to deal with skewed asset returns since they do not account for the direction of the asymmetry. Adding expected loss to MAD, i.e. using the risk measure $\text{MAD} + \Eb$, considerably impacts the allocation and tilts the weights in accordance with expected returns (and hence skewness). When we look at Expected Shortfall at $\alpha=0.95$, we observe a larger impact of skewness compared to the two previous symmetrical deviation measures. Factoring out expected loss from Expected Shortfall  (i.e. considering $\text{ES}_{0.95} - \Eb$) does not significantly impact the portfolio allocation because such an expected loss is very small relative to the large losses that Expected Shortfall captures. Using spectral risk measures $\rho_{h}$ and $\rho_{h} - \Eb$ yield portfolios which are similar to those obtained with Expected Shortfall. Similarly, the use of the extreme variantile $\upsilon_{0.99}$ allows to capture skewness risk.

\section*{Conclusion}
In this paper, we provide an analysis of the Risk Budgeting problem. First, we provide mathematical results that prove the existence of a unique solution to the Risk Budgeting problem. Then, in light of the rising interest for constructing Risk Budgeting portfolios for Expected Shortfall instead of volatility, we show that such a task can be performed using gradient descent tools when a mixture of multivariate Student-t distributions is assumed for asset returns. More generally, in model-based or model-free settings, this is still the case using stochastic gradient descent and by exploiting a variational characterization of Expected Shortfall. Beyond Expected Shortfall, we show that the Risk Budgeting problem actually boils down to a stochastic optimization problem for a wide range of popular risk measures. We provide numerical results that validate our theoretical findings and discuss the computational advantage associated with the stochastic optimization viewpoint introduced in this paper.

\section*{Acknowledgment}
Warmful thanks go to Stan Uryasev for a discussion on risk measures and quadrangles. The paper was presented at Université Paris 1 Panthéon-Sorbonne, Oxford University, King's College and University of Bologna during seminars or conferences. Interactions with the audience in these seminars helped us improve the paper. In particular, we were contacted by the authors of~\cite{freitas2022risk} following our talk at King's College who developed related ideas at the same time.

\section*{Data Availability Statement}

The data that support the findings of this study are available from the corresponding author upon reasonable request.

\newpage
\bibliographystyle{plain}

\begin{thebibliography}{10}

\bibitem{acerbi2002spectral}
Carlo Acerbi.
\newblock Spectral measures of risk: A coherent representation of subjective
  risk aversion.
\newblock {\em Journal of Banking \& Finance}, 26(7):1505--1518, 2002.

\bibitem{acerbi2002coherence}
Carlo Acerbi and Dirk Tasche.
\newblock On the coherence of expected shortfall.
\newblock {\em Journal of Banking \& Finance}, 26(7):1487--1503, 2002.

\bibitem{artzner1999coherent}
Philippe Artzner, Freddy Delbaen, Jean-Marc Eber, and David Heath.
\newblock Coherent measures of risk.
\newblock {\em Mathematical Finance}, 9(3):203--228, 1999.

\bibitem{barzilai1988two}
Jonathan Barzilai and Jonathan~M Borwein.
\newblock Two-point step size gradient methods.
\newblock {\em IMA journal of numerical analysis}, 8(1):141--148, 1988.

\bibitem{ben1986optimized}
Aharon Ben-Tal and Adi Ben-Israel.
\newblock Optimized certainty equivalents for decisions under uncertainty.
\newblock Technical report, Texas Univ. at Austin, Center For Cybernetic
  Studies, 1986.

\bibitem{ben2007old}
Aharon Ben-Tal and Marc Teboulle.
\newblock An old-new concept of convex risk measures: The optimized certainty
  equivalent.
\newblock {\em Mathematical Finance}, 17(3):449--476, 2007.

\bibitem{best1991sensitivity}
Michael~J Best and Robert~R Grauer.
\newblock On the sensitivity of mean-variance-efficient portfolios to changes
  in asset means: some analytical and computational results.
\newblock {\em The Review of Financial Studies}, 4(2):315--342, 1991.

\bibitem{black1992global}
Fischer Black and Robert Litterman.
\newblock Global portfolio optimization.
\newblock {\em Financial Analysts Journal}, 48(5):28--43, 1992.

\bibitem{bruder2016risk}
Benjamin Bruder, Nazar Kostyuchyk, and Thierry Roncalli.
\newblock Risk parity portfolios with skewness risk: An application to factor
  investing and alternative risk premia.
\newblock {\em Available at SSRN 2813384}, 2016.

\bibitem{bruder2012managing}
Benjamin Bruder and Thierry Roncalli.
\newblock Managing risk exposures using the risk budgeting approach.
\newblock {\em Available at SSRN 2009778}, 2012.

\bibitem{bun2017cleaning}
Jo{\"e}l Bun, Jean-Philippe Bouchaud, and Marc Potters.
\newblock Cleaning large correlation matrices: tools from random matrix theory.
\newblock {\em Physics Reports}, 666:1--109, 2017.

\bibitem{choueifaty2008toward}
Yves Choueifaty and Yves Coignard.
\newblock Toward maximum diversification.
\newblock {\em The Journal of Portfolio Management}, 35(1):40--51, 2008.

\bibitem{choueifaty2013properties}
Yves Choueifaty, Tristan Froidure, and Julien Reynier.
\newblock Properties of the most diversified portfolio.
\newblock {\em Journal of Investment Strategies}, 2(2):49--70, 2013.

\bibitem{contduguestcoignard}
Rama Cont, Romain Deguest, and Xue Dong He.
\newblock Loss-based risk measures.
\newblock {\em Statistics \& Risk Modeling}, 30(2):133--167, 2013.

\bibitem{demey2010risk}
Paul Demey, S{\'e}bastien Maillard, and Thierry Roncalli.
\newblock Risk-based indexation.
\newblock {\em Available at SSRN 1582998}, 2010.

\bibitem{demiguel2009optimal}
Victor DeMiguel, Lorenzo Garlappi, and Raman Uppal.
\newblock Optimal versus naive diversification: How inefficient is the 1/n
  portfolio strategy?
\newblock {\em The Review of Financial Studies}, 22(5):1915--1953, 2009.

\bibitem{embrechts2021bayes}
Paul Embrechts, Tiantian Mao, Qiuqi Wang, and Ruodu Wang.
\newblock Bayes risk, elicitability, and the expected shortfall.
\newblock {\em Mathematical Finance}, 31(4):1190--1217, 2021.

\bibitem{freitas2022risk}
Bernardo Freitas Paulo~da Costa, Silvana~M Pesenti, and Rodrigo Targino.
\newblock Risk budgeting portfolios from simulations.
\newblock {\em Available at SSRN 4038514}, 2022.

\bibitem{gava2021turning}
J{\'e}r{\^o}me Gava, Francisco Guevara, and Julien Turc.
\newblock Turning tail risks into tailwinds.
\newblock {\em The Journal of Portfolio Management}, 47(4):41--70, 2021.

\bibitem{gneiting2011making}
Tilmann Gneiting.
\newblock Making and evaluating point forecasts.
\newblock {\em Journal of the American Statistical Association},
  106(494):746--762, 2011.

\bibitem{griveau2013fast}
Th{\'e}ophile Griveau-Billion, Jean-Charles Richard, and Thierry Roncalli.
\newblock A fast algorithm for computing high-dimensional risk parity
  portfolios.
\newblock {\em Available at SSRN 2325255}, 2013.

\bibitem{hyndman1996sample}
Rob~J Hyndman and Yanan Fan.
\newblock Sample quantiles in statistical packages.
\newblock {\em The American Statistician}, 50(4):361--365, 1996.

\bibitem{laloux1999noise}
Laurent Laloux, Pierre Cizeau, Jean-Philippe Bouchaud, and Marc Potters.
\newblock Noise dressing of financial correlation matrices.
\newblock {\em Physical Review Letters}, 83(7):1467, 1999.

\bibitem{ledoit2004honey}
Olivier Ledoit and Michael Wolf.
\newblock Honey, i shrunk the sample covariance matrix.
\newblock {\em The Journal of Portfolio Management}, 30(4):110--119, 2004.

\bibitem{lemperiere2017risk}
Yves Lemperiere, Cyril Deremble, Trung-Tu Nguyen, Philip Seager, Marc Potters,
  and Jean-Philippe Bouchaud.
\newblock Risk premia: Asymmetric tail risks and excess returns.
\newblock {\em Quantitative Finance}, 17(1):1--14, 2017.

\bibitem{lezmi2018portfolio}
Edmond Lezmi, Hassan Malongo, Thierry Roncalli, and Rapha{\"e}l Sobotka.
\newblock Portfolio allocation with skewness risk: A practical guide.
\newblock {\em Available at SSRN 3201319}, 2018.

\bibitem{lintner1965security}
John Lintner.
\newblock Security prices, risk, and maximal gains from diversification.
\newblock {\em The Journal of Finance}, 20(4):587--615, 1965.

\bibitem{maillard2010properties}
S{\'e}bastien Maillard, Thierry Roncalli, and J{\'e}r{\^o}me Te{\"\i}letche.
\newblock The properties of equally weighted risk contribution portfolios.
\newblock {\em The Journal of Portfolio Management}, 36(4):60--70, 2010.

\bibitem{markowitz1952}
Harry Markowitz.
\newblock Portfolio selection.
\newblock {\em The Journal of Finance}, 7(1):77--91, 1952.

\bibitem{mossin1966equilibrium}
Jan Mossin.
\newblock Equilibrium in a capital asset market.
\newblock {\em Econometrica: Journal of the Econometric Society}, pages
  768--783, 1966.

\bibitem{rietz1988equity}
Thomas~A Rietz.
\newblock The equity risk premium a solution.
\newblock {\em Journal of Monetary Economics}, 22(1):117--131, 1988.

\bibitem{rockafellar2000optimization}
R~Tyrrell Rockafellar and Stanislav Uryasev.
\newblock Optimization of conditional value-at-risk.
\newblock {\em Journal of Risk}, 2:21--42, 2000.

\bibitem{rockafellar2002conditional}
R~Tyrrell Rockafellar and Stanislav Uryasev.
\newblock Conditional value-at-risk for general loss distributions.
\newblock {\em Journal of Banking \& Finance}, 26(7):1443--1471, 2002.


\bibitem{rockafellar2002deviation}
R~Tyrrell Rockafellar, Stanislav~P Uryasev, and Michael Zabarankin.
\newblock Deviation measures in risk analysis and optimization.
\newblock {\em University of Florida, Department of Industrial \& Systems
  Engineering Working Paper}, (2002-7), 2002.

\bibitem{rockafellar2006generalized}
R~Tyrrell Rockafellar, Stan Uryasev, and Michael Zabarankin.
\newblock Generalized deviations in risk analysis.
\newblock {\em Finance and Stochastics}, 10(1):51--74, 2006.

\bibitem{rockafellar2008risk}
R~Tyrrell Rockafellar, Stan Uryasev, and Michael Zabarankin.
\newblock Risk tuning with generalized linear regression.
\newblock {\em Mathematics of Operations Research}, 33(3):712--729, 2008.

\bibitem{rockafellar2013fundamental}
R~Tyrrell Rockafellar and Stan Uryasev.
\newblock The fundamental risk quadrangle in risk management, optimization and
  statistical estimation.
\newblock {\em Surveys in Operations Research and Management Science},
  18(1-2):33--53, 2013.

\bibitem{rockafellar2018superquantile}
R~Tyrrell Rockafellar and Johannes O. Royset.
\newblock Superquantile/CVaR risk measures: Second-order theory.
\newblock {\em Annals of Operations Research},
  262:3--28, 2018.

\bibitem{roncalli2013introduction}
Thierry Roncalli.
\newblock {\em Introduction to risk parity and budgeting}.
\newblock CRC Press, 2013.

\bibitem{roncalli2014introducing}
Thierry Roncalli.
\newblock Introducing expected returns into risk parity portfolios: A new
  framework for asset allocation.
\newblock {\em Available at SSRN 2321309}, 2014.

\bibitem{sharpe1964capital}
William~F Sharpe.
\newblock Capital asset prices: A theory of market equilibrium under conditions
  of risk.
\newblock {\em The Journal of Finance}, 19(3):425--442, 1964.

\bibitem{spinu2013algorithm}
Florin Spinu.
\newblock An algorithm for computing risk parity weights.
\newblock {\em Available at SSRN 2297383}, 2013.

\bibitem{tasche2007capital}
Dirk Tasche.
\newblock Capital allocation to business units and sub-portfolios: the {E}uler
  principle.
\newblock {\em arXiv preprint arXiv:0708.2542}, 2007.

\bibitem{tobin1958liquidity}
James Tobin.
\newblock Liquidity preference as behavior towards risk.
\newblock {\em The Review of Economic Studies}, 25(2):65--86, 1958.

\bibitem{Treynor1962JackT}
Jack~L Treynor.
\newblock Toward a theory of market value of risky assets.
\newblock {\em Econometric Modeling: Capital Markets}, 1962.

\bibitem{wang2020risk}
Ruodu Wang and Yunran Wei.
\newblock Risk functionals with convex level sets.
\newblock {\em Mathematical Finance}, 30(4):1337--1367, 2020.

\end{thebibliography}

\newpage
\appendix
\section{Description of the different procedures to evaluate the performances of SGD, OSBGD and MSBGD by simulation}
\label{appendixA}


\medskip
\begin{itemize}
\item We define $\text{DGP}_{\text{true}}$ that can generate $d-$dimensional return vectors. Since it is not convenient and realistic to manually define the parameters of $\text{DGP}_{\text{true}}$, as the number of parameters rapidly increases with respect to $d$, we use the following systematic approach: for a given $d$, we randomly select $d$ stocks from the S\&P500 Index components (in April 2022) and fit a  mixture of two multivariate Student-t distributions using the expectation-maximization algorithm on daily return data (August 2008--April 2022) with fixed degrees of freedom ($\nu_1=4.0$ and $\nu_2=2.5$). The estimated values of $p$, $\mu_1$, $\mu_2$, $\Lambda_1$ and $\Lambda_2$ are then used as the parameters of $\text{DGP}_{\text{true}}$.
\item We draw the sample $\mathcal{X}_{hist}$ of size $n=3500$ from $\text{DGP}_{\text{true}}$. 
\item From $\mathcal{X}_{hist}$, we propose two main approaches to compute the Risk Budgeting portfolio:
\begin{itemize}
\item Model-free approach, where we rely on
\begin{itemize}
\item SGD to solve Problem~\eqref{stocRB}, using $\mathcal{X}_{hist}$;
\item OSBGD to solve Problem~\eqref{empRB}, using $\mathcal{X}_{hist}$ to compute the gradient at each iteration.
\end{itemize}
\item Model-based approach, where we first estimate a model using $\mathcal{X}_{hist}$ to define $\text{DGP}_{\text{est}}$. Then, invoke
\begin{itemize}
\item SGD to solve Problem~\eqref{stocRB}, using a sample $\mathcal{X}_{sim}$ of size $10^6$ drawn from $\text{DGP}_{\text{est}}$;
\item OSBGD to solve Problem~\eqref{empRB}, using always the same sample $\mathcal{X}_{sim}$ of size $10^6$ initially drawn from $\text{DGP}_{\text{est}}$ to compute the gradient at each iteration;
\item MSBGD to solve Problem~\eqref{empRB}, using a new sample $\mathcal{X}_{sim}$ of size $10^5$ drawn from $\text{DGP}_{\text{est}}$ repeatedly to compute the gradient at each iteration.\\
\end{itemize}
In the model-based approach, estimating $\text{DGP}_{\text{est}}$ using $\mathcal{X}_{hist}$ is the key step. Ideally, the ultimate goal is to get a model as close as possible to $\text{DGP}_{\text{true}}$. In this paper, we consider three different situations:
\begin{itemize}
\item we perfectly estimate the model, i.e. $\text{DGP}_{\text{est}}$ is formally equivalent to $\text{DGP}_{\text{true}}$;
\item we correctly specify the family of the $X$ distribution which means that we fit a  mixture of two multivariate Student-t distributions applying the expectation-maximization algorithm on $\mathcal{X}_{hist}$ with fixed degrees of freedom ($\nu_1=4.0$ and $\nu_2=2.5$); the estimated models is denoted $\text{DGP}_{\text{SM}}$; 
\item we do not correctly specify the family of the $X$ distribution and fit a  mixture of two multivariate Gaussian distributions applying the expectation-maximization algorithm on $\mathcal{X}_{hist}$; the estimated model is denoted $\text{DGP}_{\text{GM}}$.
\end{itemize}
\end{itemize}
\end{itemize}

\end{document}